\documentclass{article}

\usepackage[total={8.5in,11in},top=1in,left=1in,right=1in,bottom=1in]{geometry}
\usepackage{mathrsfs}

\usepackage{amssymb,amsfonts,amsmath,amsthm}
\usepackage{graphicx}
\usepackage{subfigure}
\usepackage[noend]{algpseudocode}
\usepackage{algorithmicx,algorithm}
\usepackage{lineno}
\usepackage{verbatim}
 
\usepackage[textsize=footnotesize]{todonotes}
\usepackage{tabularx}
\usepackage{graphicx}
\usepackage{enumerate}
\usepackage{caption}
\newtheorem{theorem}{Theorem}[section]
\newtheorem{lemma}[theorem]{Lemma}

\newtheorem{corollary}[theorem]{Corollary}

\newtheorem{remark}{Remark}

\newcommand{\bigoh}{\mathcal{O}}
\usepackage{microtype}
\newif\iflong
\newif\ifshort

\longtrue
\iflong
\else
\shorttrue
\fi

\newcommand{\Nat}{\mathbb{N}}

\usepackage{tikz}
\usetikzlibrary{calc}
\usetikzlibrary{decorations.shapes}
\tikzset{
    ncbar angle/.initial=90,
    ncbar/.style={
        to path=(\tikztostart)
        -- ($(\tikztostart)!#1!\pgfkeysvalueof{/tikz/ncbar angle}:(\tikztotarget)$)
        -- ($(\tikztotarget)!($(\tikztostart)!#1!\pgfkeysvalueof{/tikz/ncbar angle}:(\tikztotarget)$)!\pgfkeysvalueof{/tikz/ncbar angle}:(\tikztostart)$)
        -- (\tikztotarget)
    },
    ncbar/.default=0.5cm,
}

\tikzset{square left brace/.style={ncbar=0.2cm}}
\tikzset{square right brace/.style={ncbar=-0.2cm}}

\tikzset{round left paren/.style={ncbar=0.5cm,out=120,in=-120}}
\tikzset{round right paren/.style={ncbar=0.5cm,out=60,in=-60}}

\tikzset{decorate sep/.style 2 args={decorate,decoration={shape backgrounds,shape=circle,shape size=#1,shape sep=#2}}}

\def\boxit#1{\vbox{\hrule\hbox{\vrule\kern4pt
  \vbox{\kern1pt#1\kern1pt}
\kern2pt\vrule}\hrule}}

\newcommand{\NP}{\mbox{{\sf NP}}}
\newcommand{\FPT}{\mbox{{\sf FPT}}}
\newcommand{\W}{{\cal W}}

\begin{document}

\title{\Large Near-Optimal Algorithms for Point-Line Covering Problems}

\author{
Jianer Chen\thanks{Department of Computer Science and Engineering, Texas A\&M University, College Station,  TX 77843, USA. Email: {\tt chen@cse.tamu.edu.}} \and 
Qin Huang\thanks{Department of Computer Science and Engineering,
    Texas A\&M University, College Station,  TX 77843, USA. Email: {\tt huangqin@tamu.edu.}}
    
\and Iyad Kanj\thanks{School of Computing, DePaul University, Chicago, IL 60604, USA. Email: {\tt ikanj@cs.depaul.edu}.} 
  
    \and Ge Xia\thanks{Department of Computer Science, Lafayette College, Easton, PA 18042, USA. Email: {\tt xiag@lafayette.edu.}} 
    }

\date{}

\maketitle





\thispagestyle{empty} 
\begin{abstract} 

We study fundamental point-line covering problems in computational geometry, in which the input is a set $S$ of points in the plane. The first is the {\sc Rich Lines} problem, which asks for the set of all lines that each covers at least $\lambda$ points from $S$, for a given integer parameter $\lambda \geq 2$; this problem subsumes the {\sc 3-Points-on-Line} problem and the {\sc Exact Fitting} problem, which---the latter---asks for a line containing the maximum number of points. The second is the \NP-hard problem {\sc Line Cover}, which asks for a set of $k$ lines that cover the points of $S$, for a given parameter $k \in \Nat$. Both problems have been extensively studied. In particular, the {\sc Rich Lines} problem is a fundamental problem whose solution serves as a building block for several algorithms in computational geometry.

For {\sc Rich Lines} and {\sc Exact Fitting}, we present a randomized Monte Carlo algorithm that achieves a lower running time than that of Guibas et al.'s algorithm [{\it Computational Geometry} 1996], for a wide range of the parameter $\lambda$. We derive lower-bound results showing that, for $\lambda =\Omega(\sqrt{n \log n})$, the upper bound on the running time of this randomized algorithm matches the lower bound that we derive on the time complexity of {\sc Rich Lines} in the algebraic computation trees model.

For {\sc Line Cover}, we present two kernelization algorithms: a randomized Monte Carlo algorithm and a deterministic algorithm. Both algorithms improve the running time of existing kernelization algorithms for {\sc Line Cover}. We derive lower-bound results showing that the running time of the randomized algorithm we present comes close to the lower bound we derive on the time complexity of kernelization algorithms for {\sc Line Cover} in the algebraic computation trees model. 

\end{abstract}

\section{Introduction}
We study fundamental problems in computational geometry pertaining to covering a set $S$ of $n$ points in the plane with lines. The first problem, referred to as {\sc Rich Lines}, is defined as:

\vspace*{-3mm}
\begin{quote}
{\sc Rich Lines}: Given a set $S$ of $n$ points and an integer parameter $\lambda \geq 2$, compute the set of lines that each covers at least $\lambda$ points.
\end{quote}


A special case of {\sc Rich Lines} that has received attention is the {\sc Exact Fitting} problem~\cite{Guibas1996}, which asks for computing a line that covers the maximum number of points in $S$. {\sc Exact Fitting} subsumes the well-known {\sc 3-Points-on-Line} problem in an obvious way.  

 The {\sc Rich Lines} problem is a fundamental problem whose solution serves as a building block for several algorithms in computational geometry~\cite{bends,castro,iyad,Grantson2006,efficient,things}, including algorithms for the fundamental {\sc Line Cover} problem, which is our other focal problem: 
\vspace*{-3mm}
\begin{quote}
    {\sc Line Cover}: Given a set $S$ of $n$ points and a parameter $k \in \Nat$, decide if there exist at most $k$ lines that cover all points in $S$.
\end{quote}


\iflong
See Figure~\ref{fig:workflowedge} for an illustration of {\sc Rich Lines}, {\sc Exact Fitting}, and {\sc Line Cover}. 

\begin{figure}[!h]
  \centering
    \includegraphics[width=0.4\textwidth]{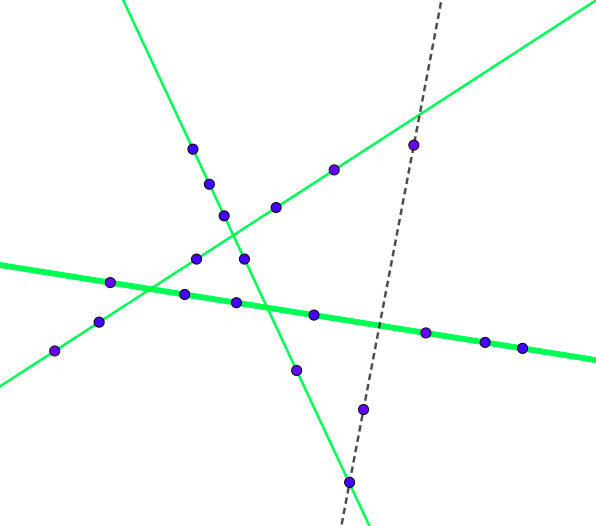}
  \caption{Illustration of an instance of {\sc Line Cover} with $k = 4$, an instance of {\sc Rich Lines} with $\lambda = 5$, and an instance of {\sc Exact Fitting}. The set of all lines is the solution to the {\sc Line Cover} instance; the set of solid green lines is the solution to the {\sc Rich Lines} instance; and the bold green line is a solution to the {\sc Exact Fitting} instance.}
  \label{fig:workflowedge}
\end{figure}

\fi
The {\sc Line Cover} problem is \NP-hard~\cite{megiddo}, and has been extensively studied in parameterized complexity~\cite{Afshani2016,cao,castro,Grantson2006,kratsch,things,wangjx}, especially with respect to kernelization. Guibas et al.'s algorithm~\cite{Guibas1996} for {\sc Rich Lines} was used to give a simple  kernelization algorithm that computes a kernel of size\footnote{In this paper, the size of the kernel for {\sc Line Cover} stands for the number of points in the kernel.}
 at most $k^2$, and this upper bound on the kernel size was proved to be essentially tight by Kratsch et al.~\cite{kratsch}.
 
The current paper derives both upper and lower bounds on the time complexity of {\sc Rich Lines}, and the time complexity of the kernelization of {\sc Line Cover}. Most of the algorithmic upper-bound results we present are randomized Monte Carlo algorithms, providing guarantees on the running time of the algorithms, but may make one-sided errors with a small probability. Our work is motivated by the applications of both problems to on-line data analytics~\cite{bigdataalgorithm}, where massive data processing within a guaranteed time upper bound is required (e.g., dynamic or streaming environments~\cite{alman,spaa2015,bigdataalgorithm}). In such settings, where the data set has an enormous size, classical algorithmic techniques become infeasible, and timely pre-processing the very large input in order to reduce its size becomes essential. Therefore, we seek algorithms whose running time is nearly linear and whose space complexity is low, trading off the optimality/correctness of the algorithm with a small probability.

\subsection{Related Work}
Both {\sc Rich Lines} and {\sc Exact Fitting} were studied by Guibas et al.~\cite{Guibas1996}, motivated by their applications in statistical analysis (e.g., linear regressions), computer vision, pattern recognition, and computer graphics~\cite{app1,app2}. Guibas et al.~\cite{Guibas1996} developed an $\bigoh(\min\{\frac{n^2}{\lambda}\ln \frac{n}{\lambda}, n^2\})$-time algorithm for {\sc Rich Lines}, and used it to solve {\sc Exact Fitting} within the same time upper bound.  Guibas et al.'s algorithm~\cite{Guibas1996}  was subsequently used in many algorithmic results~\cite{bends,castro,iyad,Grantson2006,efficient,things} pertaining to geometric covering problems and their applications. 

The {\sc Line Cover} problem has been extensively studied with respect to several computational frameworks, including approximation~\cite{cao,Grantson2006} and parameterized complexity~\cite{Afshani2016,cao,castro,kratsch,things,wangjx}. 
The problem is known to be APX-hard~\cite{kumar} and is approximable within ratio $\log n$, being a special case of the set cover problem~\cite{johnson}. 

From the parameterized complexity perspective, several fixed-parameter tractable algorithms for {\sc Line Cover} were developed~\cite{Afshani2016,Grantson2006,things,wangjx}, leading to the current-best algorithm that runs in time $(c \cdot k/\log k)^k n^{\bigoh(1)}$~\cite{Afshani2016}, for some constant $c > 0$. Guibas et al.'s algorithm~\cite{Guibas1996} was used in several works to give a kernel of size $k^2$ that is computable in time $\bigoh(\min \{\frac{n^2}{k} \log(\frac{n}{k}),n^2\})$~\cite{castro,kratsch,things}. 
This quadratic kernel size was shown to be essentially tight by Kratsch et al.~\cite{kratsch}, who showed that: For any $\epsilon > 0$, unless the polynomial-time hierarchy collapses to the third level, {\sc Line Cover} has no kernel of size $\bigoh(k^{2-\epsilon})$. 

Kernelization algorithms for {\sc Line Cover} have drawn attention in recent research in massive-data processing; Mnich~\cite{bigdataalgorithm} discusses how the {\sc Line Cover} problem is used in such settings, where the point set represents a 
very large collection of observed (accurate) data, and the solution sought is a model consisting of at most $k$ linear predictors~\cite{freedman}.

Chitnis et al.~\cite{spaa2015} studied {\sc Line Cover} in the streaming model, and Alman et al.~\cite{alman} studied the problem in the dynamic model.  We mention that Chitnis et al.'s streaming algorithm~\cite{spaa2015} may be used to give a Monte Carlo kernelization algorithm for {\sc Line Cover} running in time $\bigoh(n (\log n)^{\bigoh(1)})$, and the dynamic algorithm of Alman et al.~\cite{alman} may be used to give a deterministic kernelization algorithm for {\sc Line Cover} running in time $\bigoh(n k^2)$.

We finally note that there has been considerable work on randomized algorithms for geometric problems (see~\cite{agarwal,clarkson1,clarkson3,randomsampling}, to name a few). The most relevant of which to our work is the randomized algorithm for approximating geometric set covering problems~\cite{bronnimann,clarkson4} (see also~\cite{agarwal}), which implies an $\bigoh(\log k)$-factor approximation algorithm for the optimization version of {\sc Line Cover} whose expected running time is $\bigoh(nk (\log n) (\log k))$.

\subsection{Results and Techniques}
In this paper, we develop new tools to derive upper and lower bounds on the time complexity of {\sc Rich Lines} and the kernelization time complexity of {\sc Line Cover}. Our results and techniques are summarized as follows.

\subsubsection{Results for {\sc Rich Lines}} We present a randomized one-sided errors Monte Carlo algorithm for {\sc Rich Lines} that, with probability at least $1-\frac{3}{n^2}$, returns the correct solution set, where $n$ is the number of points. The algorithm achieves a lower running time upper bound than Guibas et al.'s algorithm~\cite{Guibas1996} for a wide range of the parameter $\lambda$, namely for $\lambda= \Omega(\log n)$, and matches its running time otherwise. For instance,  when $\lambda =\Theta(\sqrt{n \log n})$, the running time of our algorithm is $\bigoh(n \log n)$, whereas that of Guibas et al.'s algorithm is $\bigoh(n^{3/2} \sqrt{ \log n})$, yielding a $(\sqrt{n/\log n})$-factor improvement. We show that, for $\lambda =\Omega(\sqrt{n \log n})$, the upper bound of $\bigoh(n \log(\frac{n}{\lambda}))$ on the running time of our randomized algorithm matches the lower bound that we derive on the time complexity of the problem in the algebraic computation trees model\ifshort ~($\spadesuit$)\fi. The algorithm for {\sc Rich Lines} implies an algorithm for {\sc Exact Fitting} with the same performance guarantees---as shown by Guibas et al.~\cite{Guibas1996}, obtained by binary-searching for the value of $\lambda$ that corresponds to the line(s) containing the maximum number of points. 

The crux of the technical contributions leading to the randomized algorithm we present is a set of new tools we develop pertaining to point-line incidences and sampling. The aforementioned tools allow us to show that, by sampling a smaller subset of the original set of points, with high probability, we can reduce the problem of computing the set of $\lambda$-rich lines in the original set to that of computing the set of $\lambda'$-rich lines in the smaller subset, where $\lambda'$ is a smaller parameter than $\lambda$. 

The time lower-bound result we present\ifshort ~($\spadesuit$)\fi~is obtained via a 2-step reduction. The first employs Ben-Or's framework~\cite{benor} to show a time lower bound of $\Omega(n \log(\frac{n}{\lambda}))$ in the algebraic computation trees model on a problem that we define, referred to as the {\sc Multiset Subset Distinctness} problem.  We then compose this reduction with a reduction from {\sc Multiset Subset Distinctness} to {\sc Rich Lines}, thus establishing the $\Omega(n \log(\frac{n}{\lambda}))$ lower-bound result for {\sc Rich Lines}. We note that these reductions are very ``sensitive'', and hence need to be crafted carefully, as the lower-bound results apply for \emph{every} value of $n$ and $\lambda$.

\subsubsection{Results for {\sc Line Cover}} 
We derive a lower bound on the time complexity of kernelization algorithms for {\sc Line Cover} in the algebraic computation trees model and show that the running time of any such algorithm must be at least $c n\log k$ for some constant $c > 0$\ifshort ~($\spadesuit$)\fi; more specifically, one cannot asymptotically improve either of the two factors $n$ or $\log k$ in this term. This result particularly rules out the possibility of a kernelization algorithm that runs in $\bigoh(n)$ time (i.e., in linear time).   We derive this lower bound by combining a lower-bound result by Grantson and Levcopoulos~\cite{Grantson2006} on the time complexity of {\sc Line Cover} with a result that we prove in this paper connecting the time complexity of {\sc Line Cover} to its kernelization time complexity. In fact, it is not difficult to develop a kernelization algorithm for {\sc Line Cover} that runs in time $\bigoh(n \log k + g(k))$ for some computable function $g(k)$, and computes a kernel of size $k^2$. This can be done by processing the input in ``batches'' of size roughly $k^2$ each; this is implied by the algorithm in~\cite{Grantson2006}, which runs in time $\Omega(n \log k + k^4 \log k)$, and approximates the optimization version of {\sc Line Cover}.  Therefore, we focus on 
developing kernelization algorithms where the function $g(k)$ in their running time is as small as possible. Since we can assume that $n \geq k^2$ (otherwise, the instance is already kernelized), we may assume that $g(k) =\Omega(k^2 \log k)$. Therefore, we endeavor to develop a kernelization algorithm for which the function $g(k)$---in its running time---is as close as possible to $\bigoh(k^2 \log k)$, and hence, a kernelization algorithm whose running time is as close as possible to $\bigoh(n \log k)$. In addition, reducing the function $g(k)$ serves well our purpose of obtaining near-linear-time kernelization algorithms for {\sc Line Cover} due to their potential applications~\cite{freedman,bigdataalgorithm}. 

We present two kernelization algorithms for {\sc Line Cover}. The first is a randomized one-sided errors Monte Carlo algorithm that runs in time $\bigoh(n\log k+ k^2(\log^2 k)(\log\log k)^2)$ and space $\bigoh(k^2 \log^2 k)$ and, with probability at least $1-\frac{2}{k^3}$, computes a kernel of size at most $k^2$.  The second is a deterministic algorithm that computes a kernel of size at most $k^2$ in time $\bigoh(n\log k+k^3(\log^3 k) (\sqrt{\log\log k}))$. 
Both algorithms improve the running time of existing kernelization algorithms for {\sc Line Cover}~\cite{alman,spaa2015,castro,Grantson2006,kratsch,things}. Moreover, the running time of the randomized algorithm comes within a factor of $\log k(\log\log k)^2$ from the derived lower bound on the time complexity of kernelization algorithms for {\sc Line Cover}.  
 
The key tool leading to the improved kernelization algorithms is partitioning the ``saturation range'' of the saturated lines (i.e., the lines that each contains at least $k+1$ points and must be in the solution) in the batch of points under consideration into intervals, thus defining a spectrum of saturation levels. Then the algorithm for {\sc Rich Lines} (either the randomized or Guibas et al.'s algorithm~\cite{Guibas1996}) is invoked starting with the highest saturation threshold, and iteratively decreasing the threshold until either: the saturated lines computed cover ``enough'' points of the batch under consideration, or the total number of saturated lines computed is ``large enough'' thus making ``enough progress'' towards computing the line cover. This scheme enables us to amortize the running time of the algorithm that computes the saturated lines, creating a win/win situation and improving the overall running time.  

\iflong
\subsection{Organization of the Paper} Section~\ref{sec:prelim} briefly introduces the necessary terminologies and background. Section~\ref{sec:exactfitting} contains the results for {\sc Rich Lines} and Section~\ref{sec:linecover} contains those for {\sc Line Cover}. Section~\ref{sec:lowerbounds} contains the lower-bound results.
Section~\ref{sec:conclusion} contains concluding remarks. 
\fi

\section{Preliminaries}
\label{sec:prelim}
We assume familiarity with basic geometry, probability, and parameterized complexity and refer to the following standard textbooks on some of these subjects~\cite{Cygan2015,fptbook,flumgrohe,upfal,niedermeier}. For a positive integer $i$, we write $[i]$ for $\{1,2,\ldots,i\}$. We write ``\emph{w.h.p.}''~as an abbreviation for ``with high probability'', and we write ``\emph{u.a.r.}''~as an abbreviation for ``uniformly at random''.

{\bf Probability.} The \emph{union bound} states that, for any probabilistic events $E_1, E_2, \ldots, E_j$, we have:
$\Pr\left (\bigcup_{i=1}^{j} E_i \right)\le \sum_{i=1}^j \Pr(E_i)$.  For any discrete random variables
$X_1, \ldots, X_n$ with finite expectations, it is well known that: $E[\sum_{i=1}^n X_i]=\sum_{i=1}^n E[X_i]$, where $E[X]$ denotes the expectation of $X$.

\begin{theorem}[Theorem 4 in~\cite{wassily}] \label{the:neg}
Let $\mathscr{C}=\{x_1, \ldots, x_N\}$, where $x_i\in \{0,1\}$ for $i\in[N]$. 
Let $X_1, X_2, \ldots, X_j$ denote a 
random sample without replacement from $\mathscr{C}$ and let $Y_1, Y_2, \ldots, Y_j$ 
denote a random sample with replacement from $\mathscr{C}$. 
If the function $f(x)$ is continuous and convex then $E[f(\sum_{i=1}^j X_i)] \le E[f(\sum_{i=1}^j Y_i)]$. 
Moreover, $E[e^{h(\sum_{i=1}^j X_i)}] \le E[e^{h(\sum_{i=1}^j Y_i)}]$, where $h$ is a constant. 
\end{theorem}

The following lemma, for the sum of a random 
sample without replacement from a finite set, can be viewed as an application of 
Chernoff's bounds---customized to our needs---to negatively correlated random variables:

\begin{lemma} \ifshort {\rm ($\spadesuit$)}\fi \label{new-chernoff}
Let $\mathscr{C}=\{x_1, \ldots, x_N\}$, where $x_i\in \{0,1\}$ for $i\in[N]$. 
Let $X_1, X_2, \ldots, X_j$ denote a 
random sample without replacement from $\mathscr{C}$. 
Let $X=\sum_{i=1}^j X_i$,
$\mu=E[X]$, 
and $\mu_1, \mu_2$ be any two values such that $\mu_1\le \mu \le \mu_2$. Then,
(A) for any $\delta>0$, we have
$\Pr(X\ge (1+\delta)\mu_2) \le \left ( \frac{e^{\delta}}{(1+\delta)^{(1+\delta)}} \right)^{\mu_2}$; and (B) for any $0<\delta<1$, we have
$\Pr(X\le (1-\delta)\mu_1) \le \left ( \frac{e^{-\delta}}{(1-\delta)^{(1-\delta)}} \right)^{\mu_1}$.
\end{lemma}
\iflong
\begin{proof}
This proof proceeds exactly as the proof of the Chernoff bound~\cite{upfal} with minor changes. 
Let $Y_1, Y_2, \ldots, Y_j$ 
denote a random sample with replacement from $\mathscr{C}$.
Note that $Y_1, \ldots, Y_j$ are independent and identically distributed, i.e., $p=\Pr(Y_i=1)=$ 
$\frac{x_1+\cdots x_N}{N}$, and $\Pr(Y_i=0)=1-p$. 
Let  $Y=\sum_{i=1}^j Y_i$. Clearly, $E[X]=E[Y]=jp=\mu$. 
Applying Markov's inequality, for any $t>0$ we have 
\begin{eqnarray}
    \Pr(X \ge (1+\delta)\mu_2) &=& \Pr(e^{tX} \ge e^{t(1+\delta)\mu_2}) \nonumber\\     
                                           &\le&  \frac{E[e^{tX}]}{e^{t(1+\delta)\mu_2}}   \label{CE1}   \\
                                           &\le&    \frac{E[e^{tY}]}{e^{t(1+\delta)\mu_2}}    \label{CE2}   \\
                                           &=&  \frac{\prod_{i=1}^j E[e^{tY_i}]}{e^{t(1+\delta)\mu_2}},  \label{CE3}
\end{eqnarray}
where inequality (\ref{CE2}) is obtained from (\ref{CE1}) by Theorem \ref{the:neg}, and 
inequality (\ref{CE3}) is derived from (\ref{CE2}) because $Y_1, \ldots, Y_j$ are independent and 
identically distributed. Since $\Pr(Y_i=1)=p$ and $\Pr(Y_i=0)=1-p$, we have 
$E[e^{tY_i}]=pe^t+(1-p)=1+p(e^t-1) \le e^{p(e^t-1)}$, where in the last inequality we have used 
the fact that, for any $y$, $1+y \le e^y$. Plugging this into equality (\ref{CE3}), we have:
\begin{eqnarray}
    \Pr(X \ge (1+\delta)\mu_2) &\le&  \frac{e^{\mu(e^t-1)}}{e^{t(1+\delta)\mu_2}} \nonumber  \\
                                               &\le& \frac{e^{\mu_2(e^t-1)}}{e^{t(1+\delta)\mu_2}},  \label{CE4}
\end{eqnarray}
where inequality (\ref{CE4}) is obtained because $\mu\le \mu_2$ and $t>0$. It is not difficult to verify that the 
function (\ref{CE4}) attains its minimum value of $(\frac{e^{\delta}}{(1+\delta)^{1+\delta}})^{\mu_2}$ by 
setting $t=\ln (1+\delta)>0$. Hence, 
$$\Pr(X \ge (1+\delta)\mu_2) \le \left (\frac{e^{\delta}}{(1+\delta)^{1+\delta}} \right)^{\mu_2}.$$

Similarly, we prove that for any $0<\delta<1$,
$\Pr(X\le (1-\delta)\mu_1) \le \left ( \frac{e^{-\delta}}{(1-\delta)^{(1-\delta)}} \right)^{\mu_1}$ holds.
For any $t<0$, applying the Markov's inequality, we have
\begin{eqnarray}
    \Pr(X \le (1-\delta)\mu_1) &=& \Pr(e^{tX} \ge e^{t(1-\delta)\mu_1}) \nonumber\\     
                                           &\le&  \frac{E[e^{tX}]}{e^{t(1-\delta)\mu_1}}   \label{CE5}   \\
                                           &\le&    \frac{E[e^{tY}]}{e^{t(1-\delta)\mu_1}}    \label{CE6}   \\
                                           &=&  \frac{\prod_{i=1}^j E[e^{tY_i}]}{e^{t(1-\delta)\mu_1}}  \label{CE7} \\
                                           &\le&  \frac{e^{\mu(e^t-1)}}{e^{t(1-\delta)\mu_1}} \label{CE8}  \\
                                           &\le& \frac{e^{\mu_1(e^t-1)}}{e^{t(1-\delta)\mu_1}}.  \label{CE9}
\end{eqnarray}

Inequality (\ref{CE6}) is obtained from (\ref{CE5}) by Theorem \ref{the:neg}. 
Inequality (\ref{CE7}) is derived from (\ref{CE6})  because $Y_1, \ldots, Y_j$ are 
independent and identically distributed. Inequality (\ref{CE8}) is obtained from (\ref{CE7}) because 
$E[e^{tY_i}]=1+p(e^t-1) \le e^{p(e^t-1)}$. Inequality (\ref{CE9}) is obtained from (\ref{CE8}) because $t<0$ and 
$\mu \ge \mu_1$. 
It is not difficult to verify that the 
function (\ref{CE9}) attains its minimum value of $(\frac{e^{-\delta}}{(1-\delta)^{1-\delta}})^{\mu_1}$ by 
setting $t=\ln (1-\delta)<0$. Hence, 
$$ \Pr(X \le (1-\delta)\mu_1) \le \left(\frac{e^{-\delta}}{(1-\delta)^{1-\delta}} \right)^{\mu_1}, $$
thus completing the proof.
\end{proof}
\fi

{\bf Point-Line Incidences.}  Let $S$ be a set of points. A line $l$ \emph{covers} a point $p \in S$ if $l$ passes through $p$ (i.e., contains $p$). A set $L$ of lines {\it covers} $S$ if every point in $S$ is covered by at least one line in $L$. A line $l$ is \emph{induced} by $S$ if $l$ covers at least 2 points of $S$, and a set $L$ of lines is \emph{induced} by $S$ if every line in $L$ is induced by $S$. For a set $L$ of lines, we define $I(L, S)$ as 
$I(L, S)=|\{(q, l)\mid q\in S \cap l, l\in L\}|$; that is, $I(L, S)$ is the number of incidences between $L$ and $S$. For a line $l$, let $I(l, S)=|\{(q, l) \mid q\in S\cap l\}|$. 
The following theorems upper bound $I(L, S)$ and the complexity of computing it:   

\begin{theorem}[\cite{comb3}] \label{IncidenceBound}
$I(L,S)\le \frac{5}{2}(mn)^{2/3}+m+n$, where $n=|S|$ and $m=|L|$.
\end{theorem}

The theorem below follows from Theorem~3.1 in~\cite{incidences} after a slight modification, as was also observed by~\cite{Hopcroft's}:

\begin{theorem}[\cite{Hopcroft's,incidences}] \label{hopcroft}
Let $S$ be a set of $n$ points and $L$ a set of $m$ lines in the plane. The set of incidences between $S$ and $L$, and hence $I(L, S)$, can be computed in (deterministic) time $\bigoh(n\log m+m\log n+(mn)^{2/3}2^{\bigoh(\log^*(n+m))})$. Moreover, within the same running time, we can compute for each line $l \in L$ the set of points in $S$ that are contained in $l$.
\end{theorem}
Let $P$ be a subset of $S$, and let $x \in \Nat$. We say that a line $l$ is \emph{$x$-rich} for~$P$ if $l$ covers at least $x$  points from $P$; when $P$ is clear from the context, we will simply say that $l$ is $x$-rich.

We will use the following result, which is implied from a more general theorem in~\cite{chazelle}, to answer point-line incidency queries:

\begin{theorem}[\cite{chazelle}] \label{point-location}
Given a set $L$ of $m$ lines in the plane, in $O(m^2)$ (deterministic) time and space we can preprocess $L$ so that, given any point $p$, we can determine if $p$ is covered by (a line in) $L$ 
in $O(\log m)$ time. 
\end{theorem}

The following extends Theorem~2 in~\cite{endre1983}, which applies only when the constant $c <1$:

\begin{theorem}\ifshort {\rm ($\spadesuit$)}\fi
\label{RichlineBound1}
Let $S$ be a set of $n$ points, let $c>0$ be a constant, and let $k$ be an integer such that $2\le k \le  c\sqrt{n}$. Let $L$ be the set of 
$k$-rich lines for~$S$. 
Then $|L| <\max\{\frac{40n^2}{k^3},\frac{40c^2n^2}{k^3}\}$. 
\end{theorem}

\iflong
\begin{proof}
Since the total number of lines, and hence the number of lines containing at least $k$ points, is at most 
${ n\choose 2}$, the statement of the theorem holds for $k=2, 3, 4$ obviously. In the following, we prove that the statement holds for 
$k\ge 5$. 

Let $x = \max \{ \frac{40n^2}{k^3}, \frac{40c^2n^2}{k^3}\}$. We proceed by contradiction, and assume that $|L| \ge x$.
 Then $L$ contains a subset $L'$ of $m=\lfloor {x} \rfloor$ lines, each containing at least $k$ points in $P$. 
Thus, $I(L', P)\le \frac{5}{2}(mn)^{2/3}+m+n$ by Theorem \ref{IncidenceBound}. Since $m=\lfloor{x} \rfloor$, 
$I(L', P)\le \frac{5}{2}(xn)^{2/3}+x+n$. Let $g(x)=\frac{5}{2}(xn)^{2/3}+x+n$. Since each line in $L'$ 
covers at least $k$ points, we have $I(L',P) \ge km>k(x-1)\ge kx-n$. Let $f(x)=kx-n$. 
Then $f(x)-g(x)=kx-n-(\frac{5}{2}(xn)^{2/3}+x+n)$ and $g(x)\ge I(L',P)>f(x)$. We will derive a contradiction by showing below that $f(x) \geq g(x)$.
We distinguish two cases: \\

{\bf Case 1:} If $c\ge 1$, then $x=\frac{40c^2n^2}{k^3}$ and we have:

\begin{alignat}{2}
    f(x)-g(x) &= kx-n-(\frac{5}{2}(xn)^{2/3}+x+n) \nonumber\\
                  &= \frac{40c^2n^2}{k^2}-n-\frac{5}{2}(\frac{40c^2n^3}{k^3})^{2/3}-\frac{40c^2n^2}{k^3}-n \label{inequ12} \\
                  &> \frac{40c^2n^2}{k^2}-2n-\frac{30c^2n^2}{k^2}-\frac{40c^2n^2}{k^3} \label{inequ13}\\
                  &= \frac{10c^2n^2}{k^2}-2n-\frac{40c^2n^2}{k^3} \nonumber \\
                  &= \frac{10c^2n^2}{k^2}(1-\frac{4}{k})-2n \label{inequ15} \\
                  &\ge \frac{10c^2n^2}{k^2}(1-\frac{4}{5})-2n \label{inequ16}\\
                  &\ge \frac{2c^2n^2}{(c\sqrt{n})^2}-2n \label{inequ17} \\
                  &= 0. \nonumber
\end{alignat}

Inequality (\ref{inequ13}) is obtained from (\ref{inequ12}) because $c\ge 1$.
Inequality (\ref{inequ16}) is obtained from (\ref{inequ15}) because $k\ge 5$ by assumption. 
Inequality (\ref{inequ17}) is derived from (\ref{inequ16}) because $k\le c\sqrt{n}$. 
This contradicts the fact that $f(x) < g(x)$. 

{\bf Case 2:} If $c<1$, then $x=\frac{40n^2}{k^3}$.  Similarly, we have:
\begin{eqnarray}
    f(x)-g(x) &=& kx-n-(\frac{5}{2}(xn)^{2/3}+x+n) \nonumber\\
                  &=& \frac{40n^2}{k^2}-n-\frac{5}{2}(\frac{40n^3}{k^3})^{2/3}-\frac{40n^2}{k^3}-n \nonumber\\
                  &>& \frac{40n^2}{k^2}-2n-\frac{30n^2}{k^2}-\frac{40n^2}{k^3} \nonumber \\
                  &=& \frac{10n^2}{k^2}-2n-\frac{40n^2}{k^3} \nonumber \\
                  &=& \frac{10n^2}{k^2}(1-\frac{4}{k})-2n 
                  \label{inequ113}\\
                  &\ge& \frac{10n^2}{k^2}(1-\frac{4}{5})-2n \label{inequ114} \\
                  &\ge & \frac{2n^2}{(\sqrt{n})^2}-2n 
                  \label{inequ115}\\
                  &=& 0.  \nonumber
\end{eqnarray}
Inequality (\ref{inequ114}) is obtained from (\ref{inequ113}) since $k\ge5$. Inequality (\ref{inequ115}) is obtained from (\ref{inequ114}) since 
$k\le \sqrt{n}$.  This contradicts the fact that $f(x) < g(x)$, and proves the theorem.
\end{proof}

\fi

{\bf Parameterized Complexity.}
A {\it parameterized problem} $Q$ is a subset of $\Omega^* \times
\mathbb{N}$, where $\Omega$ is a fixed alphabet. Each instance of $Q$ is a pair $(I, \kappa)$, where $\kappa \in \Nat$ is called the {\it
parameter}. A parameterized problem $Q$ is
{\it fixed-parameter tractable} (\FPT), if there is an
algorithm, called an {\em \FPT-algorithm},  that decides whether an input $(I, \kappa)$
is a member of $Q$ in time $f(\kappa) \cdot |I|^{\bigoh(1)}$, where $f$ is a computable function and $|I|$ is the input instance size.  The class \FPT{} denotes the class of all fixed-parameter
tractable parameterized problems. A parameterized problem is {\em kernelizable}
if there exists a polynomial-time reduction that maps an instance $(I, \kappa)$ of
the problem to another instance $(I', \kappa')$ such that (1) $|I'| \leq f(\kappa)$ and $\kappa' \leq f(\kappa)$, where $f$ is a computable function, and (2) $(I,\kappa)$ is a yes-instance
of the problem if and only if $(I',\kappa')$ is. The instance
$(I',\kappa')$ is called the {\em kernel} of~$I$.

\section{A Randomized Algorithm for {\sc Rich Lines}}
\label{sec:exactfitting}
In this section, we present a randomized algorithm for {\sc Rich Lines} that achieves a better running time than Guibas et al.'s algorithm~\cite{Guibas1996} for a wide range of the parameter $\lambda$ (in the problem definition). \iflong We will show in Section~\ref{sec:lowerbounds}\fi ~\ifshort We show ($\spadesuit$) \fi that for $\lambda =\Omega(\sqrt{n \log n})$, the upper bound of $\bigoh(n \log(\frac{n}{\lambda}))$ on the running time of our randomized algorithm for {\sc Rich Lines} matches the lower bound on its time complexity that we derive in the algebraic computation trees model. 

\ifshort
We first present an intuitive low-rigor explanation of the randomized algorithm and the techniques entailed. The crux of the technical results in this section lies in Lemma~\ref{l1-bound2}. This lemma shows that, by sampling a smaller subset $S'$ of $S$ whose size depends on $\lambda$, w.h.p.~we can reduce the problem of computing the set of $\lambda$-rich lines for~$S$ to that of computing the set of $\lambda'$-rich lines for~$S'$, where $\lambda' < \lambda$. 

The algorithm exploits the above technical results as follows. Given an instance $(S, \lambda)$ of {\sc Rich Lines}, the algorithm samples a subset $S' \subseteq S$ whose size depends on $\lambda$.  Depending on the value of $\lambda$, the algorithm defines a threshold value $\lambda'$, and computes the set $L'$ of $\lambda'$-rich lines for~$S'$. As we show in this section, w.h.p.~$L'$ contains all the $\lambda$-rich lines for~$S$, and hence, by sifting through the lines in $L'$, the algorithm computes the solution to $(S, \lambda)$. 
\fi

\iflong
We first present an intuitive low-rigor explanation of the randomized algorithm and the techniques entailed, which may serve as a roadmap for navigating through this section.

Lemma~\ref{RichlineBound2} is an incidence structural lemma that will be used---together with Theorem~\ref{RichlineBound1} in Section~\ref{sec:prelim}---to upper bound the number of $\lambda$-rich lines in an instance $(S, \lambda)$ of {\sc Rich Lines}.  
Lemma \ref{RichlineBound2} and Theorem \ref{RichlineBound1} are employed for proving the key Lemma~\ref{l1-bound2}. 

The crux of the technical results in this section lies in Lemma~\ref{l1-bound2}. Intuitively speaking, this lemma shows that, by sampling a smaller subset $S'$ of $S$ whose size depends on $\lambda$, then w.h.p.~we can reduce the problem of computing the set of $\lambda$-rich lines for~$S$ to that of computing the set of $\lambda'$-rich lines for~$S'$, where $\lambda' < \lambda$.

The algorithm exploits the above technical results as follows. Given an instance $(S, \lambda)$ of {\sc Rich Lines}, the algorithm samples a subset $S' \subseteq S$ whose size depends on $\lambda$. Depending on the value of $\lambda$, the algorithm defines a threshold value $\lambda'$, and computes the set $L'$ of $\lambda'$-rich lines for~$S'$. As we show in this section, w.h.p.~$L'$ contains all the $\lambda$-rich lines for~$S$, and hence, by sifting through the lines in $L'$, the algorithm computes the solution to the instance $(S, \lambda)$. 
\fi

\begin{lemma} \ifshort {\rm ($\spadesuit$)} \fi
\label{RichlineBound2}
Let $\lambda \ge 2\sqrt{n}$. The number of 
$\lambda$-rich lines for~$S$ is at most $\frac{2n}{\lambda}$. 
\end{lemma}
\iflong
\begin{proof}
Let $L=\{l_1, l_2, ..., l_m\}$ be the set of $\lambda$-rich lines for~$P$. 
Denote by $z_i$, for $i=1,2,\ldots, m$, the number of points covered by $l_i$. The set $L$ covers at least $\sum_{i=1}^m z'_i$ points, where $z'_i=\max\{z_i-i+1, 0\}$. 
This is true since $l_1$ covers at least $z_1$ points, $l_2$ covers at least $z_2-1$ new points
(excluding at most 1 point on $l_1$), and in general, $l_i$ covers at least $z_i-i+1$ new points
(excluding at most $i-1$ points covered by $\{l_1, ..., l_{i-1}\}$) if $i\le z_i$ or 0 otherwise. 
Suppose, to get a contradiction, that $m> \frac{2n}{\lambda}$, and consider the first $m'=\lceil {\frac{2n}{\lambda}} \rceil$ 
lines. Then $\sum_{i=1}^m z'_i \ge \sum_{i=1}^{m'}z'_i =\sum_{i=1}^{m'} (z_i-i+1)$ since 
$i\le m' < z_i$. Hence, $\sum_{i=1}^{m'}z'_i\ge \lambda \cdot m' - \frac{(m'-1)m'}{2}$
$\ge 2n - \frac{(\frac{2n}{\lambda}+1)\frac{2n}{\lambda}}{2} = 2n-\frac{n}{\lambda}-\frac{2n^2}{\lambda^2}$. Since $\lambda \ge 2\sqrt{n}$, $\sum_{i=1}^{m'}z'_i \ge 2n -\frac{\sqrt{n}}{2} - \frac{n}{2} > n$ (assuming w.l.o.g.~that $n > 1$), which is a contradiction. 
Therefore, we have $m \le \frac{2n}{\lambda}$. 
\end{proof}
\fi

Throughout this section, $S$ denotes a set of $n \geq 3$ points. Let $S'(m)$ be a set formed by sampling without replacement $m\le n$ points from $S$ uniformly and independently at random.

\begin{lemma}\label{l1-bound2}  
Let $\lambda$ be an integer satisfying $140\ln^{3/2} n \le \lambda \le  n$. Let $S'(m)$ be as defined above where $m=\lceil{\frac{140 n\ln n}{\lambda}} \rceil$. Let $L_1$ be the set of $\lambda$-rich lines for~$S$, and let $L_3$ be the set of $(98\ln n)$-rich lines for~$S'(m)$. Then, with probability at least 
$1-\frac{2}{n^2}$, we have: 
$(1)$ $L_1 \subseteq L_3$, and $(2)$
if $\lambda \ge 5\sqrt{n}$, 
$|L_3| \le \frac{5n}{\lambda}$; 
if $\lambda <  5\sqrt{n}$, $|L_3|<\frac{2500n^2}{\lambda^3}$. 
\end{lemma}

\begin{proof}
Let $L_2$ be the set of lines induced by $S$ containing less than $\frac{2\lambda}{5}$ points each.
Without loss of generality, suppose that $L_1=\{l_1, \ldots, l_d \}$ and $L_2=\{l_{d+1}, \ldots, l_{d'}\}$, and note that $d \le d' \le { n \choose 2}$.
For $i \in [d']$, let $x_i$ be the number of points in $S$ covered by $l_i$, and 
let $X'_i$ be the random variable, where $X'_i$ is the number of the points in $S'(m)$ on $l_i$.

For $i \in [d]$, we have $E[X'_i]=\frac{x_i}{n}\cdot m \ge 140\ln n$ since $x_i \ge \lambda$.
Applying part (B) of the Chernoff bounds in Lemma~\ref{new-chernoff} with $\mu_1=140\ln n$, we have 
$\Pr(X'_i \le (1-3/10)\cdot 140\ln n) \le \left ( \frac{e^{-3/10}}{(1-3/10)^{1-3/10}}  
\right)^{140\ln n} \le \frac{1}{n^4}$, where the last inequality can be easily verified by a simple analysis.
Let $E_i$, for $i\in[d]$, denote the event that $X'_i \le (1-3/10)\cdot 140\ln n$. Applying the union bound, we have $\Pr(\bigcap_{i=1}^{d} \overline{E_i}) = 1-\Pr( \bigcup_{i=1}^d E_i)  \ge 1-d \cdot \frac{1}{n^4} \ge 1-\frac{1}{n^2}$. 
Let ${\cal E}=\bigcap_{i=1}^{d} \overline{E_i}$. The probability that every line $l_i \in L_1$ contains at least $98\ln n$  points of $S'(m)$ is at least  $1-\frac{1}{n^2}$. That is to say, with probability at least $1-\frac{1}{n^2}$, we have
$L_1\subseteq L_3$.

For $i=d+1, \ldots, d'$, $E[X'_i]=\frac{x_i}{n}\cdot m<\frac{x_i}{n}(\frac{140n\ln n}{\lambda}+1) 
\le 56\ln n+\frac{2\lambda}{5n}<57\ln n$, since $x_i\le \frac{2\lambda}{5}$ and $140 \ln^{3/2} n \le \lambda \le n$. 
Applying part (A) of the Chernoff bounds in Lemma~\ref{new-chernoff} with $\mu_2=57\ln n$, we get 
$\Pr(X'_i \ge (1+\frac{13}{19})\cdot 57\ln n ) \le$ $\left( \frac{e^{13/19}}{(1+13/19)^{1+13/19}} 
\right)^{57\ln n} \le \frac{1}{n^4}$, where the last inequality can be easily verified by a simple analysis.  
Consequently, via the union bound, the probability that every line $l_i \in L_2$ contains less than $96\ln n$ sampled points  is at least $1-(d'-d) \cdot \frac{1}{n^4} \ge  1-\frac{1}{n^2}$. 
It follows that, with probability at least $1-\frac{1}{n^2}$, we have $L_3\cap L_2=\emptyset$.  

Altogether, with probability at least 
$1-\frac{2}{n^2}$, $L_1\subseteq L_3$ and $L_2\cap L_3=\emptyset$.
 Recall that each line in $L_3$ covers at least $\frac{2\lambda}{5}$ points of $S$.
 If $\frac{2\lambda}{5} \ge 2\sqrt{n}$, i.e., $\lambda \ge 5\sqrt{n}$, we have $|L_3|\le \frac{2n}{2\lambda/5}=\frac{5n}{\lambda}$ by Lemma \ref{RichlineBound2}. If $\frac{2\lambda}{5} < 2\sqrt{n}$, i.e. $\lambda < 5\sqrt{n}$, by Theorem~\ref{RichlineBound1} (with $c\le 2$), we have $|L_3|<\frac{40\cdot 2^2n^2}{(2\lambda/5)^3}=\frac{2500n^2}{\lambda^3}$.
It follows that, with probability at least $1-\frac{2}{n^2}$, parts (1) and (2) of the lemma hold.
\end{proof}

\begin{algorithm}
     \textbf{Input:} a set of points $S$ and $\lambda \in \Nat$.    \\
     \textbf{Output:} The set $L$ of $\lambda$-rich lines for $S$.
      \begin{algorithmic}[1]
             \State {\bf if} {$\lambda <\ln n$ } {\bf then} apply Guibas et al.'s algorithm~\cite{Guibas1996} to compute $L$ and return $L$; 
                  \label{fitting-1}
             \State sample $x=\lceil{\frac{10n^2\ln n}{\lambda^2}}\rceil$ pairs of points $(p_1,q_1), \ldots, (p_{x}, q_{x})$
              u.a.r. from ${S \choose 2}$;   \label{fitting1}
              \State let $l_i$ be the line formed by $(p_i,q_i)$, for $i \in [x]$; let $Q_1$ be the multi-set 
              $\{l_1, l_2, \ldots, l_{x}\}$, 
              and let $Q_2$ be the set of distinct lines in $Q_1$;
              \label{fitting2}
              \State {\bf if} {$\lambda \le 140\ln^{3/2} n$ } {\bf then}  let $L=\{l\in Q_2 \mid I(l, S)\ge \lambda\}$; return $L$;  \label{fitting5}
              \State let
               $m=\lceil{\frac{140 n\ln n}{\lambda}}\rceil$, $y=98\ln{n}$;  \label{fitting6}
              \State sample $m$ points u.a.r.~from $S$ without replacement to obtain $S'(m)$;
                \label{fitting7}
             \State {\bf if} {$\lambda < 5\sqrt{n}$ } {\bf then}
               \State ~~~~ let $z=\frac{2500n^2}{\lambda^3}$; 
             \State {\bf else} let $z=\frac{5n}{\lambda}$;
             \State let $L'=\{l\in Q_2 \mid I(l, S'(m))\ge y\}$;  \label{fitting8}
             \State {\bf if} {$|L'| \le z$} {\bf then} let $L=\{l \in L'| I(l, S) \ge \lambda \}$; return $L$;    \label{fitting9}
             \State {\bf else} return $\emptyset$; \label{fitting11}
      \end{algorithmic}
      \caption{{\bf : Alg-RichLines}$(S, \lambda)$--A randomized algorithm for computing all $\lambda$-rich lines} \label{Fitting}
\end{algorithm}

Refer to {\rm\bf Alg-RichLines} for the terminologies used in the subsequent discussions.

\begin{lemma} \label{Alg-Fitting1}
Let $\lambda \in \Nat$. Let $L_1$ be the set of $\lambda$-rich for~$S$. Then, with probability at least $1-\frac{3}{n^2}$, {\rm\bf Alg-RichLines}$(S, \lambda)$ returns a set $L=L_1$. 
\end{lemma}
\begin{proof}
If $\lambda <\ln n$ then $L=L_1$ with probability 1 by Step \ref{fitting-1}, as Guibas et al.'s algorithm~\cite{Guibas1996} computes $L_1$ deterministically.

Now consider the case that $\lambda \ge \ln n$. Let $l$ be an arbitrary line in $L_1$. 
Step \ref{fitting1} samples $x$ pairs of points that determine $x$ lines. In a single sampling, 
the probability $\rho$ that $l$ is sampled is  ${\lambda \choose 2}/{n \choose 2} \ge \frac{\lambda^2}{2n^2}$ 
because $l$ covers at least $\lambda$ points of $P$. Thus, $\Pr(l\notin Q_1)= (1-\rho)^{x}\le e^{-\rho x} \le \frac{1}{n^{5}}$.
Since $|L_1|<n^2$, applying the union bound, we get $Pr(L_1 \subseteq Q_1) \geq 1-\frac{1}{n^3}$. 
Hence, we have $\Pr(L_1 \subseteq Q_2) \ge 1-\frac{1}{n^3}$ since $Q_2$ is obtained from 
$Q_1$ by removing repeated lines. Let $L_3$ be the set of $y$-rich lines 
for $S'(m)$. If $\lambda < 140\ln^{3/2} n$, then since $L_1 \subseteq Q_2$ with probability at least $1-\frac{1}{n^3}$, the algorithm returns in Step \ref{fitting5} a set $L$ that, with probability at least $1-\frac{1}{n^3}$, is equal to $L_1$ .


Finally, if $140\ln^{3/2} n \le  \lambda$, then by Lemma \ref{l1-bound2}, $L_1\subseteq L_3$ and $|L_3|\le z$ with probability at least $1-\frac{2}{n^2}$. 
Since $L'=L_3 \cap Q_2$, $|L'|\le z$ holds with probability at least $1-\frac{2}{n^2}$. Since $\Pr(L_1\subseteq Q_2)\ge 1-\frac{1}{n^3}$, we have $L_1 \subseteq (L_3\cap Q_2)$ and $|L'|\le z$ with probability at least $1-\frac{3}{n^2}$. Thus, $\Pr(L_1 \subseteq L') \ge 1-\frac{3}{n^2}$. Therefore, the algorithm returns in Step \ref{fitting9} a set $L$ equal to $L_1$ with probability at least  $1-\frac{3}{n^2}$.
\end{proof}

\ifshort
The correctness of {\rm\bf Alg-RichLines}, stated in the following theorem, follows from Lemma~\ref{Alg-Fitting1}. Its run-time analysis is tedious but straightforward. \fi

\begin{theorem} \ifshort {\rm ($\spadesuit$)} \fi
\label{thm:exactfitting}
Let $S$ be a set of $n$ points and $\lambda \in \mathbb{N}$. With probability at least $1 -\frac{3}{n^2}$, {\rm\bf Alg-RichLines}$(S, \lambda)$ solves the {\sc Rich Lines} problem. Moreover, the running time of the algorithm is:
 
\begin{enumerate}
\item [(1)] $\bigoh(n^2)$ if $\lambda< \ln n$; and
\item [(2)]   $\bigoh(n\log \frac{n}{\lambda}+\frac{n^2\log n\log \frac{n}{\lambda}}{\lambda^2})$ otherwise.
\end{enumerate}
\end{theorem}
 
\iflong
\begin{proof}
By Lemma~\ref{Alg-Fitting1}, with probability at least $1 -\frac{3}{n^2}$, {\bf Alg-RichLines}$(P, \lambda)$ correctly returns the set of $\lambda$-rich lines for~$P$. We discuss next the running time of the algorithm. \\

\noindent {\bf Case 1. $\lambda<\ln n$.} In this case the running time of the algorithm is that of Guibas et al.'s algorithm~\cite{Guibas1996}, which is $\bigoh(n^2)$. \\

It is easy to see that Step~\ref{fitting1} takes $\bigoh(x)=\bigoh(\frac{n^2\ln n}{\lambda^2})$ time.   Step~\ref{fitting2} can be implemented by sorting the slopes of the $x$ lines, which takes $\bigoh(x\ln{x})=O (\frac{n^2\ln n}{\lambda^2} \cdot (\ln \frac{n}{\lambda}+\ln\ln n))$ time. 
Step \ref{fitting7} takes time $\bigoh(m)=\bigoh(\frac{140n\ln n}{\lambda})=\bigoh(n)$ since $n \ge \lambda\ge 140\ln^{3/2} n$. 
Steps 5, 7, 8, 9 and \ref{fitting11} take constant time. 
Note that all the above running times (for Steps 2, 3, 5, \ref{fitting7}, 7, 8, 9, \ref{fitting11}) are dominated by the running time listed in item (2) of the theorem. \\


We discuss the running time of Step~\ref{fitting5} in Case 2 below, and that of Step~\ref{fitting8}  and Step~\ref{fitting9} in both Cases 3 and 4. Note that, to determine the set of rich lines for~$S$ in Steps \ref{fitting5}, \ref{fitting8} and \ref{fitting9}, we apply Theorem~\ref{hopcroft} to compute the number of points in $S$ (or $S'(m)$)
on each of the lines in question, thus determining the set of rich lines for $S$ (or $S'(m)$).

\noindent {\bf Case 2. $\ln n \le \lambda \le 140\ln^{3/2} n$.} In this case, Step \ref{fitting5} takes time $T_1=\bigoh(x\log n+n\log x+
(nx)^{2/3}2^{\bigoh(\log^* (x+n))})$ by Theorem \ref{hopcroft}. Substituting $x=\lceil{\frac{10n^2\ln n}{\lambda^2}}\rceil$, we obtain
$T_1=\bigoh(\frac{n^2\ln^2 n}{\lambda^2}+n\ln n+\frac{n^2\ln^{2/3}n}{\lambda^{4/3}}2^{\bigoh(\log^*n)})=\bigoh(\frac{n^2\ln^2 n}{\lambda^2})$, which is dominated by the running time listed in item (2) of the theorem.

We discuss the running time of Step \ref{fitting8} in both Case 3 and 4. 
Note that $|S'(m)| =m=\lceil{\frac{140n\ln n}{\lambda}}\rceil$ and $|Q_2| \le x=\bigoh(\frac{n^2\ln n}{\lambda^2})$. 
By Theorem \ref{hopcroft}, Step~\ref{fitting8} takes time: 
\begin{eqnarray*}
T_2 &=& \bigoh(m\log x+x\log m+(xm)^{2/3} 2^{\bigoh(\log^* (x+m))}) \label{NN20} \\
&=& \bigoh(\frac{n^2\ln n}{\lambda^2}(\ln (\frac{n}{\lambda})+\ln\ln n)+\frac{n^2}{\lambda^2}\ln^{4/3}(n)2^{\bigoh(\log^* n)}).
\label{NN21}
\end{eqnarray*}
If $\lambda \le n^{2/3}$, we have $\frac{n^2}{\lambda^2}\ln^{4/3}(n)2^{\bigoh(\log^* n)}=\bigoh(\frac{n^2\log n\log \frac{n}{\lambda}}{\lambda^2})$, and
if $n^{2/3}<\lambda\le n$, we have $\frac{n^2}{\lambda^2}\ln^{4/3}(n)2^{\bigoh(\log^* n)}=\bigoh(n\log \frac{n}{\lambda})$.
Altogether, $\frac{n^2}{\lambda^2}\ln^{4/3}(n)2^{\bigoh(\log^* n)}=\bigoh(n\log \frac{n}{\lambda}+\frac{n^2\log n\log \frac{n}{\lambda}}{\lambda^2})$. Therefore, 
$T_2=\bigoh(n\log \frac{n}{\lambda}+\frac{n^2\log n\log \frac{n}{\lambda}}{\lambda^2})$, which is dominated by the running 
time listed in item (2) of the theorem.

\noindent {\bf Case 3. $140\ln^{3/2} n < \lambda < 5\sqrt{n}$.} 

Step~\ref{fitting9} takes time $T_3=\bigoh(n\log |L'|+|L'|\log n+
(n|L'|)^{2/3}2^{\bigoh(\log^*(n+|L'|))})$ by Theorem \ref{hopcroft}. 
Since $|L'|\le z=\frac{2500n^2}{\lambda^3}$, we have 
$T_3=\bigoh(n\ln n+\frac{n^2\log n}{\lambda^3}+\frac{n^2}{\lambda^2}2^{\bigoh(\log^* n)})=\bigoh(n\log \frac{n}{\lambda}+\frac{n^2\log n\log \frac{n}{\lambda}}{\lambda^2})$.
Thus, the total running time in this case is $\bigoh(n\log \frac{n}{\lambda}+\frac{n^2\log n\log \frac{n}{\lambda}}{\lambda^2})$.  \\

\noindent {\bf Case 4. $\lambda \ge 5\sqrt{n}$}.

Step~\ref{fitting9} takes time 
$T_4=\bigoh(n\log |L'|+|L'|\log n+(n|L'|)^{2/3}2^{\bigoh(\log^*(n+|L'|))})$ by Theorem \ref{hopcroft}. 
Since $|L'|\le z =\bigoh(\frac{n}{\lambda})$, 
$T_4=\bigoh(n\log \frac{n}{\lambda}+\frac{n}{\lambda}\log n+\frac{n^{4/3}}{\lambda^{2/3}}2^{\bigoh(\log^*n)})=\bigoh(n\log \frac{n}{\lambda}+\frac{n^{4/3}}{\lambda^{2/3}}2^{\bigoh(\log^*n)})=\bigoh(n\log \frac{n}{\lambda})$. The last equality holds because $\lambda \ge 5\sqrt{n}$.
 Consequently, 
 the total running time in this case is $\bigoh(n\log \frac{n}{\lambda}+\frac{n^2\log n\log \frac{n}{\lambda}}{\lambda^2})$. 
\end{proof}

\fi
 
Guibas et al.'s algorithm~\cite{Guibas1996} solves the  
{\sc Rich Lines} and the {\sc Exact Fitting} problems in the plane in time $\bigoh(\min\{\frac{n^2}{\lambda}\log \frac{n}{\lambda}, n^2\})$. Theorem \ref{thm:exactfitting} is an improvement over
 Guibas et al.'s algorithm~\cite{Guibas1996} for both problems for all values of $\lambda \geq \ln n$, and for $\lambda < \ln n$ it obviously has a matching running time. In particular, for $\ln n \le \lambda \le 140\ln^{3/2} n$, 
 the improvement could be in the order of $\frac{1}{\sqrt{\log n}}$ (i.e., the running time of {\bf Alg-RichLines} 
 is a $\frac{1}{\sqrt{\log n}}$-fraction of that in~\cite{Guibas1996}); 
 for $140\ln^{3/2} n <  \lambda < 5\sqrt{n}$, the improvement could be in the order of $\frac{\log n }{\sqrt{ n}}$; and for $\lambda \ge 5\sqrt{n}$, the improvement could be in the order of $\sqrt{\frac{\log n}{n}}$.

\section{Kernelization Algorithms for {\sc Line Cover}}
\label{sec:linecover}
In this section, we present a randomized Monte Carlo kernelization algorithm for {\sc Line Cover} that employs {\bf Alg-RichLines} developed in the previous section. We also show how the tools developed in this section can be used to obtain a deterministic kernelization algorithm for {\sc Line Cover} that employs Guibas et al.'s algorithm~\cite{Guibas1996}. Both algorithms improve the running time of existing kernelization algorithms for {\sc Line Cover}. Moreover, \iflong we will show in Section~\ref{sec:lowerbounds} that \fi the running time of our randomized algorithm comes close to the lower bound that we derive on the time complexity of kernelization algorithms for {\sc Line Cover} in the algebraic computation trees model\ifshort ~($\spadesuit$)\fi. The majority of this section is dedicated to proving the following theorem:

\begin{theorem}
\label{thm:mainlc}
There is a Monte Carlo randomized algorithm, {\rm\bf Alg-Kernel}, that given an instance $(S, k)$ of {\sc Line Cover}, in time $\bigoh(n\log k+ k^2(\log^2 k)(\log\log k)^2)$,
 returns an instance $(S', k')$ such that $|S'|\le k^2$, and such that with probability at least $1-\frac{2}{k^3}$, $(S', k')$ is a kernel of $(S, k)$. More specifically: (1) if $(S, k)$ is a yes-instance of {\sc Line Cover}, then with probability at least $1-\frac{2}{k^3}$, $(S', k')$ is a yes-instance of {\sc Line Cover}; and (2) if $(S, k)$ is a no-instance of  {\sc Line Cover} then $(S', k')$ is a no-instance of {\sc Line Cover}. The space complexity of this algorithm is $\bigoh(k^2\log^2 k)$.
\end{theorem}

Let $(S, k)$ be an instance of {\sc Line Cover}.  We say that a line $l$ is \emph{saturated} w.r.t.~$S$ if it is $(k+1)$-rich for~$S$. A line $l$ is \emph{unsaturated} w.r.t.~$S$ if it is not saturated.  We start by giving an intuitive explanation of the results leading to the kernelization algorithm {\rm\bf Alg-Kernel}.

The kernelization algorithm processes the set $S$ of points in ``batches" of roughly $2k^2$ uncovered points each, and for each batch $S'$, computes the saturated lines induced by $S'$ and adds them to the (partial) solution. Since processing each batch should result in computing at least one saturated line---assuming a yes-instance, the above process iterates at most $k$ times. The main task becomes to compute the saturated lines induced by a batch efficiently. One straightforward idea is to invoke {\bf Alg-RichLines} directly with $\lambda=k+1$, which, w.h.p., computes all the saturated lines in $S'$. The drawback is that {\bf Alg-RichLines} takes time $\bigoh(k^2 \log^2 k)$
 per batch, and may result in a single saturated line, and hence in an overall running time of $\bigoh(n \log k + k^3\log^2 k)$ for the kernelization algorithm.
 
 The main technical contributions of this section lie in devising a more efficient implementation of the above kernelization scheme. The improved scheme rests on two key observations: (1) the running time of {\bf Alg-RichLines} decreases as the saturation threshold (i.e., $\lambda$) of the saturated lines sought increases; and (2) assuming that a subset of the batch $S'$ needs to be covered only by saturated
 lines, then for any $\lambda < \lambda'$, it requires more saturated lines of saturation $\lambda$---where the saturation of a rich line is the number of points on it---to cover that subset of the batch than the number of saturated lines of saturation $\lambda'$.   
 
Based on the above observations, we design an algorithm {\bf Alg-SaturatedLines} that intuitively works as follows. We first partition the saturation range into intervals, thus defining a spectrum of saturation levels. Then {\bf Alg-SaturatedLines} calls {\bf Alg-RichLines} starting with the highest saturation threshold (i.e., starting with a value of $\lambda$ defining the highest saturation interval in the spectrum), and iteratively decreasing the saturation threshold until either: (1) the saturated lines computed cover ``enough'' points of the batch $S'$, or (2) the total number of saturated lines computed for the batch $S'$ is ``large enough'', thus making enough progress towards computing the $k$ lines in the line cover of $(S, k)$. 
 
 The above scheme enables us to amortize the running time of {\bf Alg-SaturatedLines} over the number of saturated lines it computes. The main kernelization algorithm, {\bf Alg-Kernel}, then calls {\bf Alg-SaturatedLines} on each batch of $2k^2$ uncovered points. As we show in the analysis, the above scheme enables a win/win situation, yielding an overall running time of $\bigoh(n\log k+ k^2(\log^2 k)(\log\log k)^2)$.  
 
We also show that, if instead of using the randomized Monte Carlo algorithm {\bf Alg-RichLines} to compute the saturated lines we use the deterministic algorithm of Guibas et al.~\cite{Guibas1996}, the above scheme yields a deterministic kernelization algorithm for {\sc Line Cover} that runs in time 
$\bigoh(n\log k+k^3(\log^3 k) \sqrt{\log\log k})$ and computes a kernel of size at most $k^2$.

We now give an intuitive low-rigor description of the technical results leading to the kernelization algorithm. Lemma \ref{distribution} is a combinatorial result showing that either the saturated lines belonging to the highest interval in the saturation spectrum cover enough points of the batch $S'$, or there is a saturation interval in the spectrum containing a ``large enough'' number of lines. Lemma~\ref{distribution} is then employed by Lemma~\ref{lemma3} to show that, w.h.p., {\bf Alg-SaturatedLines} returns a set of saturated
lines that either covers enough points of the batch $S'$, or contains a ``large enough'' number of saturated lines.  We employ  Lemma~\ref{distribution} and use amortized analysis to upper bound the running time of {\bf Alg-SaturatedLines} w.r.t.~the number of saturated lines computed by this algorithm, which we subsequently use to upper bound the running time of {\bf Alg-SaturatedLines} in Lemma~\ref{lem:kerneltime}. Finally, Theorem~\ref{thm:mainlc} employs the above results to prove the correctness of {\bf Alg-Kernel} and upper bounds its time and space complexity.  We now proceed to the details.

In what follows let $\sigma =2k^2$, and let $S' \subseteq S$ be a subset of points such that $|S'|=\sigma$.
We want to identify a subset of saturated lines w.r.t.~$S'$. We define the following notations. Let $\epsilon=\frac{\ln\ln\ln k}{\ln k}$. For $i \in \mathbb{N}$, let $y_i=1-\frac{\ln\ln k}{\ln k}-\frac{\ln\ln\ln k}{\ln k}+i\epsilon$.  Let $r$ be the minimum integer such that $k^{y_r} \ge \frac{k}{(\ln\ln k)^2}$, and note that $r=\bigoh(\ln \ln k)$. Note that we have $y_0 < y_1 < \cdots < y_r$.

We define a sequence of intervals $I_0, \ldots, I_{r+1}$ as follows: 
$I_0=[\frac{1}{k^{1+y_0}}, 1]=[\frac{\ln k(\ln\ln k)}{k^2}, 1]$, $I_i=[\frac{1}{k^{1+y_i}}, \frac{1}{k^{1+y_{i-1}}})$, 
for $i=1, 2, \ldots, r$,
and $I_{r+1}=[\frac{(k+1)k}{\sigma(\sigma-1)}, \frac{1}{k^{1+y_r}})$. Observe that the intervals $I_0, \ldots, I_{r+1}$ are mutually disjoint, and partition the interval $[\frac{(k+1)k}{\sigma(\sigma-1)}, 1]$. It is easy to verify that the lengths of the intervals $I_1, \ldots, I_{r}$ are decreasing. \iflong See Figure~\ref{fig:fptfigure} for illustration.

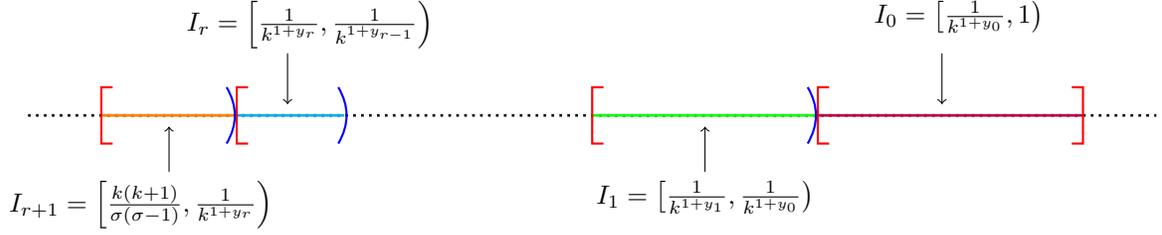
\begin{figure*}[ht]
\begin{center}
\begin{tikzpicture}[scale=0.75]
    
    \draw[dotted, color=black, line width=1] (-2,0) -- (18,0);
    
     \draw[line width=1, color=orange] (-0.7,0) -- (1.7,0);

    \draw[line width=1, color=cyan] (1.7,0) -- (3.6,0);
    
    \draw[line width=1, color=green] (8,0) -- (12,0);
    
    \draw[line width=1, color=purple] (12,0) -- (16.7,0);

    \draw [red, thick] (-0.5,-0.5) to [square left brace ] (-0.5,0.5);
    \draw [blue, thick] (1.52,-0.5) to [round right paren] (1.52,0.5);
    
    \draw [red, thick] (1.9,-0.5) to [square left brace ] (1.9,0.5);
    \draw [blue, thick] (3.5,-0.5) to [round right paren] (3.5,0.5);
    
    \draw [red, thick] (8.2,-0.5) to [square left brace ] (8.2,0.5);
    \draw [blue, thick] (11.82,-0.5) to [round right paren] (11.82,0.5);
    
    \draw [red, thick] (12.2,-0.5) to [square left brace ] (12.2,0.5);
    \draw [red, thick] (16.5,-0.5) to [square right brace] (16.5,0.5);

     \draw[shift={(0,-1)}]   node[below]{$I_{r+1} = \left[\frac{k(k+1)}{\sigma(\sigma-1)}, \frac{1}{k^{1+y_r}}\right)$};
     \draw[->](0.5,-1) -- (0.5,-0.2);
     
     \draw[shift={(3,2.2)}]   node[below]{$I_{r} = \left[\frac{1}{k^{1+y_r}}, \frac{1}{k^{1+y_{r-1}}}\right)$};
     \draw[->](2.6,1.1) -- (2.6,0.2);

     \draw[shift={(10,-1)}]   node[below]{$I_{1} = \left[\frac{1}{k^{1+y_1}}, \frac{1}{k^{1+y_{0}}}\right)$};
     \draw[->](10,-1) -- (10,-0.2);
     
     \draw[shift={(14.5,2.2)}]   node[below]{$I_{0} = \left[\frac{1}{k^{1+y_0}}, 1\right)$};
     \draw[->](14.2,1.1) -- (14.2,0.2);









%

\end{tikzpicture}
\end{center}
\caption{Illustration for the definition of the intervals  $I_0, \ldots, I_{r+1}$. 
}
\label{fig:fptfigure}

\end{figure*}
\fi

Suppose that there are $h$ saturated lines 
$l_1, \ldots, l_h$ w.r.t.~$S'$. Denote by $s_i$ the number of points in $S'$ covered by $l_i$, for $i \in [h]$. Let $\rho_i=\frac{s_i(s_i-1)}{\sigma(\sigma-1)}$, and note that $\rho_i$ belongs to one of the intervals $I_0, \ldots, I_{r+1}$. We partition the $h$ saturated lines into at most $r+2$ groups, $H'_0, \ldots, H'_{r+1}$, where $H'_i$, for $i=0, \ldots, r+1$, consists of every saturated line $l_j$, $j \in [h]$, such that $\rho_j \in I_i $. Clearly, it follows that $H'_0, \ldots, H'_{r+1}$ is indeed a partitioning of $\{l_1, \ldots, l_h\}$.

Consider {\bf Alg-SaturatedLines} for computing the saturated lines w.r.t. $S'$:

\begin{algorithm}[htbp]
     \textbf{Input:} $S'$, where $|S'|=\sigma=2k^2$; $k \in \mathbb{N}$;  and integer $r$ as defined before  \\
     \textbf{Output:} A set of points $S''$ and a set of saturated lines $L'$ 
     
      \begin{algorithmic}[1]
          
            \State {\bf for} {($i=0; i \le r+1; i++$)} {\bf do}  \label{rich7}
                    \State ~~~~{\bf if} { $i\le r$ } {\bf then} let $L'=$ {\bf Alg-RichLines}$(S', \sigma  k^{-(1+y_i)/2})$;  \label{rich8}
                    \State ~~~~{\bf else} \ let $L'=$ {\bf Alg-RichLines}$(S', k+1)$;  \label{rich88} 
                    \State ~~~~compute the set $S''\subseteq S'$ not covered by $L'$;
                    \label{rich-666}
                    \State ~~~~{\bf if} { $i=0$ and $L'$ covers at least $k^2/3$ points}
                            {\bf then} return $(S'', L')$;   \label{rich-6}
                     \State ~~~~{\bf else if} { $i\le r$  and $|L'| \ge \frac{1}{12r}k^{(1+y_{i-1})/2}$}   
                              {\bf then} return $(S'', L')$;   \label{rich10}   
                     \State ~~~~{\bf else if} {$i=r+1$ and $|L'| \ge \frac{1}{12}k^{(1+y_{r})/2}$}
                              {\bf then} return $(S'', L')$;  \label{rich12}
           \State return $(S', \emptyset)$;   \label{rich15}
      \end{algorithmic}
      \caption{{\bf : Alg-SaturatedLines}($S', k, r$)--A randomized algorithm for computing saturated lines w.r.t. $S'$} \label{richAlg}
\end{algorithm}

Now we are ready to present the kernelization algorithm, {\bf Alg-Kernel}, for {\sc Line Cover}. 

 
\begin{algorithm}[htbp]
     \textbf{Input:} $S=\{q_1, \ldots, q_n\}$; $k \in \mathbb{N}$. \\
     \textbf{Output:} an instance $(S', k')$ of {\sc Line Cover}.
      \begin{algorithmic}[1]
       \State {\bf if} {$k \leq 15$} {\bf then} return the instance $(S', k')$ described in Lemma~\ref{lem:specialcase};   \label{special}
           \State $H=\emptyset$; $S'=\emptyset$; $i=1$;           \label{ker1}
           \State construct a search structure $\Gamma_H$ for the lines in $H$ and set $\Gamma_H=\emptyset$;\label{ker22}
           \State {\bf while} { $|H|\le k$ } {\bf do}  \label{ker2}
                
                  \State ~~~~{\bf while} { $|S'|<2k^2$ and $i\leq n$ } {\bf do}\label{ker4}
                       
                      \State~~~~~~~~{\bf if} {$q_i$ is not covered by $H$ } {\bf then} add $q_i$ to $S'$ and set $i=i+1$;  \label{ker6}
                   \State ~~~~{\bf if} { $|S'|=2k^2$ } {\bf then} \label{ker7}
                        \State ~~~~~~~~let $(S', L')=$ {\bf Alg-SaturatedLines}$(S', k, r)$;  \label{ker8}
                         \State ~~~~~~~~{\bf if} {$L'=\emptyset$} {\bf then} return a (trivial) no-instance $(S', k')$; \label{ker9}
                        \State ~~~~~~~~$H=H\cup L'$; update $\Gamma_H$ for $H$; \label{ker10}
                   \State ~~~~{\bf else}
                     \State~~~~~~~~{\bf if} { $|S'| > k^2$ } {\bf then}  \label{ker12}
                        \State ~~~~~~~~~~~~$L'=$ {\bf Alg-RichLines}$(S', k+1)$; $H=H\cup L'$; \label{ker13}
                       \State ~~~~~~~~~~~~update  $\Gamma_H$ for $H$; remove the points  from $S'$ covered
                       by $L'$;  \label{ker14}
                       \State~~~~~~~~~~~~{\bf if} { $|H|>k$ or $|S'|>k^2$ } {\bf then} return a (trivial) no-instance $(S', k')$; \label{ker15} 
                  \State ~~~~~~~~return $(S', k-|H|)$;  \label{ker16}
            \State return a (trivial) no-instance $(S', k')$; \label{ker17}
      \end{algorithmic}
      \caption{{\bf : Alg-Kernel}($S, k$)--A randomized kernelization algorithm for {\sc Line Cover}} \label{Alg-PLC}
\end{algorithm}

We explain the kernelization algorithm. The kernelization algorithm works by computing w.h.p.~the set $H$ of saturated lines in $S$ and removing all points covered by these lines. Observe that, any set of more than $k^2$ points that can be covered by at most $k$ lines must contain at least one saturated line. During the execution of the algorithm, the set $S'$, which will eventually contain the kernel, contains a subset of points in $S$. We start by initializing $S'$ to the empty set, and order the points in $S$ arbitrarily. We repeatedly add the next point in $S$ (w.r.t.~the defined order) to $S'$ until either $|S'|=2k^2$, or no points are left in $S$. Afterwards, the algorithm distinguishes two cases.

If $|S'|=2k^2$, the algorithm calls {\bf Alg-SaturatedLines} to compute a subset of the saturated lines w.r.t.~$S'$. {\bf Alg-SaturatedLines} may not compute all the saturated lines in $S'$, and rather acts as a ``filtering algorithm''. This algorithm either computes a subset of saturated lines that cover at least $k^2/3$ many points in $S'$ ``efficiently'', that is more efficiently than {\bf Alg-RichLines}, which w.h.p.~computes all the saturated lines in $S'$; or computes a ``large'' set of saturated lines (a little bit less efficiently than {\bf Alg-RichLines}), thus decreasing the parameter $k$ significantly (and hence the overall execution of the algorithm). 

If $k^2 < |S'| < 2 k^2$, no more points are left in $S$ to consider. {\bf Alg-RichLines} is called at most once to compute w.h.p.~all the remaining saturated lines w.r.t.~$S'$ to return the kernel. 

We now proceed to prove the correctness and analyzing the complexity of {\bf Alg-Kernel}.
 
\begin{lemma} \ifshort {\rm ($\spadesuit$)} \fi
\label{distribution}
Given a set $S'$ of $\sigma$ points and a parameter $k$, if $S'$ can be covered by at most $k$ lines then one of the following conditions must hold: 
\begin{enumerate}
\item [(1)] $H'_0$ covers at least $\frac{\sigma-k^2}{3}$ points;  
\item [(2)] $|H'_i| \ge (\frac{\sigma-k^2}{6r\sigma})\cdot k^{(1+y_{i-1})/2}$ for some $i \in [r]$; or
\item [(3)] $|H'_{r+1}| \ge (\frac{\sigma-k^2}{6\sigma})\cdot k^{(1+y_{r})/2}$.
\end{enumerate} 
\end{lemma}

\iflong
\begin{proof}

For an arbitrary set $H'_i$, where $i \in [r+1]$, we upper bound the number of points covered by any line in $H'_i$. For a line $l_a \in H'_i$, we have 
$\frac{s_a(s_a-1)}{\sigma(\sigma-1)}<\frac{1}{k^{1+y_{i-1}}}$ by definition. It is easy to verify that, since $2 \leq s_a \leq \sigma$, we have  
$\frac{s_a(s_a-1)}{\sigma(\sigma-1)}>\frac{s_a^2}{2\sigma^2}$. 
Thus, $s_a<\frac{\sqrt{2}\sigma}{k^{(1+y_{i-1})/2}}< \frac{2\sigma}{k^{(1+y_{i-1})/2}}.$

By the assumption of the lemma, $S'$ can be covered by at most $k$ lines. Let $C$ be a set of at most $k$ lines that cover $S'$. There are at most $k-\sum_{j=0}^{r+1}|H'_j|$ unsaturated lines in $C$ that each  
covers at most $k$ points in $S'$. Thus, at least $\sigma-k^2$ points are covered by $\bigcup_{j=0}^{r+1} H'_i$. 
If $H'_0$ covers at least $\frac{\sigma-k^2}{3}$ points, part (1) of the lemma holds, and the lemma follows. Otherwise, either $\bigcup_{i=1}^r H'_i$ 
or $H'_{r+1}$ covers at least $\frac{\sigma-k^2}{3}$ points. 

If $\bigcup_{i=1}^r H'_i$ covers at least $\frac{\sigma-k^2}{3}$ points, then there exists $j \in [r]$ such that $H'_j$ covers at 
least $\frac{\sigma-k^2}{3r}$ points. For each line $l_a \in H'_j$, we have $s_a< \frac{2\sigma}{k^{(1+y_{j-1})/2}}$.
Hence, $|H'_j| \ge  \left.\frac{\sigma-k^2}{3r} \right/  \frac{2\sigma}{k^{(1+y_{j-1})/2}}=(\frac{\sigma-k^2}{6\sigma r})\cdot  k^{(1+y_{j-1})/2}$, and part (2) of the lemma, and hence the lemma, follows.

If $H'_{r+1}$ covers at least $\frac{\sigma-k^2}{3}$ points, then for every $l_a \in H'_{r+1}$, we have
$s_a< \frac{2\sigma}{k^{(1+y_{r})/2}}$. Therefore, $|H'_{r+1}| \ge \left.\frac{\sigma-k^2}{3} \right/ \frac{2\sigma}{k^{(1+y_{r})/2}}$
$=(\frac{\sigma-k^2}{6\sigma})\cdot k^{(1+y_r)/2}$, and part (3) of the lemma follows, thus completing the proof.
\end{proof}
\fi

\begin{lemma}  \label{lemma3}
Given a set $S'$ of points and a parameter $k\ge 16$, let $L'$ be the set of lines returned by {\rm\bf Alg-SaturatedLines}$(S', k, r)$. If $S'$ can be covered with at most 
$k$ lines, then with probability at least $1-\frac{1}{k^4}$ one of the following holds: 
\begin{enumerate}
\item [(1)] $L'$ covers at least $\frac{k^2}{3}$ points; 
\item [(2)] $|L'| \ge \frac{1}{12r}k^{(1+y_{i-1})/2}$ for some $i \in [r]$; or 
\item [(3)] $|L'| \ge \frac{1}{12}k^{(1+y_{r})/2}$.
\end{enumerate}
\end{lemma}

\begin{proof} 
For $i=0, \ldots , r$, let $\mathcal{H}_i=\bigcup_{j=0}^i H'_j$
and let $\mathcal{R}_i$ be the set of $(\sigma k^{-(1+y_i)/2})$-rich lines for~$S'$. Let $\mathcal{R}_{r+1}=\bigcup_{j=0}^{r+1} H'_j$. For each line $l_a \in \mathcal{H}_i$, we have $\frac{s_a(s_a-1)}{\sigma(\sigma-1)}\ge \frac{1}{k^{1+y_i}}$. Since $\frac{\sigma k^{-(1+y_i)/2}(\sigma k^{-(1+y_i)/2}-1)}{\sigma(\sigma-1)}$
$<\frac{1}{k^{1+y_i}}$, we have $s_a>\sigma k^{-(1+y_i)/2}$. Hence, $\mathcal{H}_i \subseteq 
\mathcal{R}_i$, for $i=0, \ldots, r$. 

By Theorem~\ref{thm:exactfitting}, with probability at least $1-\frac{3}{\sigma^2}>1-\frac{1}{k^4}$, Steps \ref{rich8}--\ref{rich88} will compute $L'$ such that $L'=\mathcal{R}_i$ in the $i$-th iteration.
If $S'$ can be covered with at most $k$ lines, then
by Lemma \ref{distribution} (applied with $\sigma=2k^2$), one of the following conditions must hold: 
$(i)$ $H'_0$ covers at least $\frac{k^2}{3}$ points; 
$(ii)$ $|H'_i| \ge \frac{1}{12r}k^{(1+y_{i-1})/2}$ for some $i \in [r]$; or 
$(iii)$ $|H'_{r+1}| \ge \frac{1}{12}k^{(1+y_{r})/2}$.  
If case $(i)$ holds,  then with probability at least $1-\frac{1}{k^4}$, {\rm\bf Alg-SaturatedLines} will stop at Step \ref{rich-6}. This is true since Step~\ref{rich8} computes 
$L'=\mathcal{R}_0$, and $H'_0 \subseteq \mathcal{R}_0$, and hence, $\mathcal{R}_0$ covers at least 
$k^2/3$ points. This proves part (1). If 
case $(i)$ does not hold and case $(ii)$ holds, then there exists  $i \in [r]$ such that $|H'_i| \ge \frac{1}{12r}k^{(1+y_{i-1})/2}$. It follows that there exists
$i\in [r]$ such that $|\mathcal{R}_i| \ge \frac{1}{12r}k^{(1+y_{i-1})/2}$. Hence, 
with probability at least $1-\frac{1}{k^4}$,
{\rm\bf Alg-SaturatedLines} will stop at Step \ref{rich10}
and part (2) follows.
Finally, if neither case $(i)$ nor case $(ii)$ holds, then case $(iii)$ 
must hold, which implies that $|\mathcal{R}_{r+1}| \ge \frac{1}{12}k^{(1+y_{r})/2}$. 
In that case, with probability at least $1-\frac{1}{k^4}$,  
{\rm\bf Alg-SaturatedLines} stops at Step \ref{rich12} (with $i=r+1$) and part (3) follows.
\end{proof}
 
\iflong 
The following lemma gives an upper bound on the running time of {\rm\bf Alg-SaturatedLines} amortized over the number of saturated lines it computes.

\begin{lemma} 
\label{lemma9}
Given a set $S'$ of points and a parameter $k \ge 16$, let $L'$ be the set of lines returned by {\rm\bf Alg-SaturatedLines}$(S', k, r)$. Then the following hold about the set $L'$ and the time complexity $T(k)$ of the
 algorithm: 
\begin{enumerate}
\item [(1)] if the algorithm stops at Step \ref{rich-6}, then $T(k) =\bigoh(k^2 \log k)$ and $L'$ covers at least $\frac{k^2}{3}$ points; 
\item [(2)]  if the algorithm stops at  Step \ref{rich10} with the iterator $i$ of the {\bf for loop} having value $j$, then $T(k)=\bigoh(k^{1+y_j}(\log^2 k) (\log\log k))$, 
 $|L'|\ge \frac{1}{12r}k^{(1+y_{j-1})/2}$, and\\ $T(k)/|L'|=\bigoh(k(\log^2 k)(\log\log k)^2)$; 
\item [(3)]  if the algorithm stops at Step \ref{rich12}, then $T(k)=\bigoh(k^2\log^2 k)$, $|L'| \ge \frac{1}{12}k^{(1+y_{r})/2}$, and 
$T(k)/|L'|=\bigoh(k(\log^2 k)(\log\log k))$; 
\item [(4)] otherwise, the algorithm stops at Step \ref{rich15}, and $T(k)=\bigoh(k^2\log^2 k)$. 
\end{enumerate}
\end{lemma}

\begin{proof}
To start with, let us upper bound $r$ and $k^{y_r}$.   
Recall that $y_i=1-\frac{\ln\ln k}{\ln k}-\frac{\ln\ln\ln k}{\ln k}+i\epsilon$, where $\epsilon=\frac{\ln\ln\ln k}{\ln k}$. 
Hence, $k^{y_i}=k^{1-\ln\ln k/\ln k-\ln\ln\ln k/\ln k+i\epsilon}=\frac{k^{1+i\epsilon}}{(\ln k)(\ln\ln k)}$ 
$=\frac{k}{\ln k} \cdot (\ln\ln k)^{i-1} $. 
Since $r$ is the minimum integer such that $k^{y_r} \ge \frac{k}{(\ln\ln k)^2}$, 
we have $k^{y_{r-1}}<\frac{k}{(\ln\ln k)^2}$ and $r<\frac{\ln\ln k}{\ln\ln\ln k}  =\bigoh(\ln\ln k)$. 
Moreover, since $k^{y_r}=\frac{k}{\ln k} \cdot (\ln\ln k)^{r-1}$, 
we have  
$k^{y_r}=k^{y_{r-1}}\cdot (\ln\ln k)< \frac{k}{(\ln\ln k)^2} \cdot (\ln\ln k)= \frac{k}{\ln\ln k}$.

First, observe that the running time of each iteration of the {\bf for loop} in the algorithm is dominated by the running time of the {\bf if-then-else} statement in Steps \ref{rich8}-\ref{rich88} and of Step \ref{rich-666}.
We analyze the running time of these two steps next.

In Step \ref{rich-666}, we use the result in Theorem~\ref{point-location}.  
Since $L'$ is a set of saturated lines, by Theorem~\ref{RichlineBound1}, we have 
$|L'|<\frac{40 \cdot (2k^2)^2}{k^3}=160k$ (note that $\sigma=2k^2$). 
Thus, preprocessing $L'$---needed for the result in Theorem~\ref{point-location}---takes time $\bigoh(k^2)$. Since each point-location query takes time $O(\log k)$
and $|S'|=2k^2$, computing $S''$ takes time $O(k^2\log k)$. Hence, 
Step \ref{rich-666} takes time $O(k^2\log k)$.

Next, we analyze the running time of Steps~\ref{rich8}-\ref{rich88}.  
 
(1) If the algorithm stops at Step \ref{rich-6}, then
Step \ref{rich8} is executed once throughout the algorithm. Recall that $k^{y_0}=\frac{k}{(\ln k)(\ln\ln k)}$. We have 
$\sigma k^{-(1+y_0)/2}>k>\ln (2k^2)$. The second 
inequality holds because $k\ge 16$. 
Therefore, part (2) of Theorem \ref{thm:exactfitting} applies, and 
Step~\ref{rich8} takes time $T_1(k)=\bigoh(\sigma \log (\frac{\sigma}{\sigma k^{-(1+y_0)/2}})+
\frac{\sigma^2\log \sigma}{(\sigma k^{-(1+y_0)/2})^2}\cdot \log \frac{\sigma}{\sigma k^{-(1+y_0)/2}})$. 
It follows that $T_1(k)=\bigoh(k^2\log k)$. 
Since Step \ref{rich-666} takes time $\bigoh(k^2\log k)$, we 
have $T(k)=\bigoh(k^2\log k)$ and $L'$ covers at least $\frac{k^2}{3}$ points by the condition in Step \ref{rich-6}.

(2) If the algorithm stops at Step \ref{rich10} with the iterator $i$ of the {\bf for loop} having value $j$, then from Step \ref{rich10}, we have $|L'| \ge \frac{1}{12r}k^{(1+y_{j-1})/2}$.
Moreover, Step \ref{rich8} is executed $j+1$ times. Since $\sigma k^{-(1+y_i)/2}\ge \sigma k^{-(1+y_{r})/2}$ and 
$k^{y_{r}}<\frac{k}{\ln\ln k}$, we have 
$\sigma k^{-(1+y_i)/2}>k>\ln (2k^2)$ for each $i\in \{0,1,\ldots, j\}$.
By part (2) of Theorem \ref{thm:exactfitting}, the total running time $T_2(k)$ (throughout the whole execution of the algorithm) of Step \ref{rich8} is:
\begin{eqnarray}
T_2(k) &=& \bigoh(\sum_{i=0}^j (\sigma \log  (\frac{\sigma}{\sigma k^{-(1+y_i)/2}}) + \frac{\sigma^2\log \sigma}{(\sigma k^{-(1+y_i)/2})^2} \cdot \log  (\frac{\sigma}{\sigma k^{-(1+y_i)/2}})) )  \label{A71}  \\
       &=& \bigoh(k^2(\log k) (\ln\ln k)+k^{1+y_j}(\log^2 k) (\ln\ln k))    \label{A72} \\
       &=& \bigoh(k^{1+y_j}(\log^2 k) (\log\log k)) \label{N73}.
\end{eqnarray}
The second equality above is obtained because $j\le r$ and $r=\bigoh(\ln\ln k)$. The third 
equality above is obtained because $k^{y_j} \ge k^{y_1}=\frac{k}{\ln k}$. Since Step \ref{rich-666} takes time $\bigoh(k^2\log k)$ 
and it is executed at most $r+1$ times, the total running 
time of Step \ref{rich-666} is $\bigoh((k^2\log k) \cdot r)=\bigoh(k^2(\log k)(\log\log k))$. 
Hence, $T(k)=\bigoh(k^{1+y_j}(\log^2 k) (\log\log k))$.
Since $|L'| \ge \frac{1}{12r} k^{(1+y_{j-1})/2}$, we have 
\begin{eqnarray}
T(k)/|L'| &=& \bigoh(k^{1+y_j-(1+y_{j-1})/2}(\log^2 k) (\log\log k) r) \notag \\
        &=& \bigoh(k^{1/2+y_j-y_{j-1}/2}(\log^2 k)(\log\log k)^2) \notag \\
        &=& \bigoh(k^{1/2+y_j/2+(y_j-y_{j-1})/2}(\log^2 k)(\log\log k)^2) \notag \\
        &=& \bigoh(k^{1/2+y_r/2+(y_j-y_{j-1})/2}(\log^2 k)(\log\log k)^2) \label{W711} \\
        &=& \bigoh(k^{1+(y_j-y_{j-1})/2}(\log^2 k)(\log\log k)^{3/2})   \label{W72} \\
        &=& \bigoh(k(\log^2 k)(\log\log k)^2). \label{W73}
\end{eqnarray}
Equality (\ref{W711}) is obtained because $y_j \leq y_r$.
Equality (\ref{W72}) is obtained because $k^{y_r} < \frac{k}{\ln\ln k}$. Equality (\ref{W73}) is obtained because $k^{y_i}=\frac{k}{\ln^2 k}\cdot (\ln\ln k)^i$, and hence, $k^{y_j-y_{j-1}}=\ln\ln k$.  

(3) If the algorithm stops at Step \ref{rich12}, then Step \ref{rich8} is executed $r+1$ times and 
Step \ref{rich88} is executed once. From Step~\ref{rich12}, we have $|L'| \ge \frac{1}{12}k^{(1+y_{r})/2}$. 
Again, for $i\in \{0,\ldots, r\}$, $\sigma k^{-(1+y_i)/2}>k>\ln (2k^2)$. 
By part (2) of Theorem \ref{thm:exactfitting}, 
the total running time $T_3(k)$ of Step \ref{rich8} throughout the algorithm is:
\begin{eqnarray}
T_3(k) &=& \bigoh(\sum_{i=0}^r ( \sigma \log (\frac{\sigma}{\sigma k^{-(1+y_i)/2}})+
        \frac{\sigma^2\log \sigma}{(\sigma k^{-(1+y_i)/2})^2}\cdot \log (\frac{\sigma}{\sigma k^{-(1+y_i)/2}}) ))  \label{A73}  \nonumber\\
       &=& \bigoh(k^2(\log k)(\ln\ln k)+k^{1+y_r}(\log^2 k)(\ln\ln k))    \label{A74} \nonumber\\
       &=& \bigoh(k^2\log^2 k ). \nonumber \label{A75}
\end{eqnarray}
The second equality above is obtained because $r=\bigoh(\ln\ln k)$ and the third equality is obtained 
because $k^{y_r} < k/\ln\ln k$. 
Step \ref{rich88} is executed once, and by part (2) of Theorem \ref{thm:exactfitting}, takes time $\bigoh(\sigma \log \frac{\sigma}{k} + \frac{\sigma^2\log \sigma }{k^2}\cdot \log \frac{\sigma}{k})=\bigoh(k^2\log^2 k)$. 
 
Again, since Step \ref{rich-666} takes time $\bigoh(k^2\log k)$ 
and is executed at most $r+2$ times, the total running 
time of Step \ref{rich-666} is $\bigoh((k^2\log k) \cdot r)=\bigoh(k^2(\log k)(\log\log k))$. 
Hence, $T(k)=\bigoh(k^2\log^2 k)$. Since $|L'| \ge \frac{1}{12}k^{(1+y_{r})/2}$, we have 
$T(k)/|L'|= \bigoh(k^{\frac{3-y_r}{2}}\log^2 k)$. 
Since $k^{y_r}\ge \frac{k}{(\ln\ln k)^2}$, $T(k)/|L'|=\bigoh(k(\log^2 k)(\log\log k))$.

(4) If the algorithm stops at Step \ref{rich15}, then the running time is the same as that in part (3) of this lemma, and hence is $\bigoh(k^2 \log^2 k)$.  
\end{proof}

\begin{corollary}
\label{cor:spacealgrl}
Given a point set $S'$ and a parameter $k \ge 16$, {\rm\bf Alg-SaturatedLines}$(S', k, r)$ runs in space $\bigoh(k^2\log^2 k)$. 
\end{corollary}

\begin{proof}
This follows from the fact that the space complexity of the algorithm is upper bounded by its time complexity, and that the worst-case time complexity of the algorithm occurs in Case (4) of Lemma~\ref{lemma9}, where the {\bf for loop} in Step~\ref{rich7} iterates until the end.  
\end{proof}
\fi

One technicality ensues from the definition of the saturation intervals. Since this definition entails using the term $\ln \ln \ln k$, $\ln \ln \ln k$ must be positive, and hence $k \geq 16$. This forces a separate treatment of instances in which $k \leq 15$. Note that, since $k=\bigoh(1)$, we could opt to use a brute-force algorithm in this case, or an \FPT{}-algorithm, but those would result in a polynomial running time of a higher degree than what is desired for our purpose. Instead, we provide an efficient linear-time algorithm for this special case in the following lemma:

\begin{lemma} \ifshort  {\rm ($\spadesuit$)} \fi
\label{lem:specialcase}
Given an instance $(S, k)$ of {\sc Line Cover}, where $|S|=n$ and $k \leq 15$, there is an algorithm that computes in $\bigoh(n)$ time and $\bigoh(1)$ space a kernel $(S', k')$ for $(S, k)$ such that $|S'| \leq k^2$.
\end{lemma}

\iflong
\begin{proof}
If $|S| \leq k^2$ then the instance $(S, k)$ is the desired kernel. Otherwise, initialize the solution set (of lines) $H=\emptyset$, and initialize a subset $S' \subseteq S$ to the empty set. Repeat the following process until either $|H| > k$ or until all points in $S$ have been considered.  

Add points from $S$ that are not covered by $H$ into $S'$ until $|S'|=k^2+1$. If $(S, k)$ is a yes-instance of {\sc Line Cover} then there must exist a saturated line w.r.t.~$S'$. We apply Guibas et al.'s algorithm~\cite{Guibas1996} to compute the set $L$ of saturated lines w.r.t.~$S'$. If $L = \emptyset$ or $|H| + |L| > k$, then return a trivial no-instance to {\sc Line Cover}; otherwise, add $L$ to $H$, decrease $k$ by $|L|$, and update $S'$ by removing from it all points covered by $L$.

When the above process is completed, return $(S', k)$ as the kernel. 

It is clear that the above algorithm returns a kernel of size at most $k^2$. We analyze its time and space complexity next.

Guibas et al.'s algorithm~\cite{Guibas1996} runs in time $\bigoh((|S'|^2/(k+1)) \ln (\frac{|S'|}{k+1}))=\bigoh(k^3\ln k)=\bigoh(1)$ and clearly uses $\bigoh(1)$ space. Updating $S$ can be done in $\bigoh(n)$ time and $\bigoh(1)$ space in a straightforward fashion, and all other operations can be performed in $\bigoh(1)$ time and space. Since the above process is repeated at most $k+1=\bigoh(1)$ times, the algorithm runs in $\bigoh(n)$ time and $\bigoh(1)$ space. 
\end{proof}
\fi

\begin{lemma} \ifshort {\rm ($\spadesuit$)} \fi
\label{lem:kerneltime}
Given an instance $(S, k)$ of {\sc Line Cover}, where $|S|=n$, {\rm\bf Alg-Kernel} runs in time 
$\bigoh(n\log k+ k^2(\log^2 k)(\log\log k)^2)$ and space $\bigoh(k^2\log^2 k)$.
\end{lemma}
\iflong
\begin{proof}
Step~\ref{special}  takes $\bigoh(n)$ time by Lemma~\ref{lem:specialcase}. Step \ref{ker1}, Step \ref{ker22}, Step \ref{ker9}, Step \ref{ker15}, Step \ref{ker16}, and Step \ref{ker17}  take time $\bigoh(1)$.  Step \ref{ker6}
takes time $\bigoh(\log k)$ per point of $S$. Thus, the overall running time of Steps \ref{ker4}--\ref{ker6} 
is $\bigoh(n\log k)$. 

Now we bound the total running time of Steps~\ref{ker10} and~\ref{ker14}. To implement these steps, we use the structure proposed by Chan and Nekrich~\cite{chan2018} as 
the search structure $\Gamma_H$, which, for a plane subdivision of size $n'$, uses space $\bigoh(n')$, 
and supports $\bigoh(\log n')$ (deterministic) query time and $\bigoh(\log^{1+\epsilon} n')$ (for any $\epsilon > 0$) (deterministic) update time. Before each execution of Step \ref{ker10}, $|H|\le k$ must hold. After executing 
Step \ref{ker10}, since 
$|S'|=2k^2$ before executing Step \ref{ker8}, by Theorem~\ref{RichlineBound1}, we have $|L'|< \frac{40 \cdot |S'|^2}{k^3}<160k$. Therefore, after Step \ref{ker10} is executed, we have 
$|H|<160k+k=161k$.  
For Step \ref{ker14}, 
Step \ref{ker14} is executed at most once. 
Before executing Step~\ref{ker14}, we have 
$|H|\le k$. At Step \ref{ker12}, 
since $k^2<|S'|<2k^2$ and 
$k<\sqrt{|S'|}$ hold at this point, by Theorem~\ref{RichlineBound1}, we have $|L'|< \frac{40 \cdot |S'|^2}{k^3}<160k$. Thus, $|H|<k+160k=161k$ after 
executing Step \ref{ker14}. Since $|H| =\bigoh(k)$, the structure $\Gamma_H$ supports $\bigoh(\log k)$ query time and $\bigoh(\log^{1+\epsilon} k)=\bigoh(\log^{2} k)$ update time. 
Since $|H| =\bigoh(k)$, the plane subdivision determined by the lines in $H$ and represented in $\Gamma_H$ will contain $\bigoh(k^2)$ edges when the algorithm terminates. It follows that the overall running time for constructing and updating $\Gamma_H$ (i.e., Steps~\ref{ker10} and~\ref{ker14}) throughout the algorithm is $\bigoh(k^2 \log^2 k)$.
At Step \ref{ker14}, removing the points in $S'$ covered by $L'$ takes time $\bigoh(k^2\log k)$. 
Hence, the total running time of Step \ref{ker10} 
and \ref{ker14} is $\bigoh(k^2\log^2 k)$.

Next, we bound the running time of Step \ref{ker8}. Note that at most one call to
{\bf Alg-SaturatedLines}
$(S', k, r)$ could return $L'=\emptyset$, because {\bf Alg-Kernel} would stop if the returned set $L'=\emptyset$; 
moreover, by Lemma \ref{lemma9}, the running time incurred in such call is $\bigoh(k^2\log^2 k)$. Therefore, we can focus now on the cases where \\ {\bf Alg-SaturatedLines}$(S',k, r)$ returns $L'$ such that $L' \neq \emptyset$. 

Each call to {\bf Alg-SaturatedLines} that returns a set $L' \neq \emptyset$ must return a set $L'$ satisfying: (1) $L'$ covers at least $k^2/3$ points; (2) $|L'| \ge \frac{1}{12r}k^{(1+y_{j-1})/2}$; or (3) $|L'| \ge \frac{1}{12}k^{(1+y_{r})/2}$. This is true due to Lemma~\ref{lemma9}: if the set $L'$ returned by {\bf Alg-SaturatedLines} is not empty (Case (4) of the lemma), then either $L'$ covers at least $k^2/3$ points (Case (1) of the lemma), 
 $|L'| \ge \frac{1}{12r}k^{(1+y_{j-1})/2}$ (Case (2) of the lemma), or  $|L'| \ge \frac{1}{12}k^{(1+y_{r})/2}$
(Case (3) of the lemma).

Consider the overall running time of the calls to {\bf Alg-SaturatedLines}$(S',k, r)$ in which {\bf Alg-SaturatedLines} returns $L'$ that covers at least $k^2/3$ points. By Lemma \ref{lemma9}, each such call to {\bf Alg-SaturatedLines}$(S', k, r)$ takes time $\bigoh(k^2\log k)=\bigoh(|S'| \log k)$. 
Since each call removes at least $k^2/3=\Omega(|S'|)$ points, it follows that the overall running time for these calls is $\bigoh(n\log k)$. 
For the running time of the calls to {\bf Alg-SaturatedLines}$(S',k, r)$ which result in $|L'| \ge \frac{1}{12r}k^{(1+y_{j-1})/2}$, note that there can be at most $(k+1)/|L'|$ such calls.
By Case (2) of Lemma \ref{lemma9}, the running time of each such call is $\bigoh(|L'| \cdot k(\log^2 k)(\log\log k)^2)$. 
Therefore, the overall running time for these calls is $\bigoh(k^2(\log^2 k)(\log\log k)^2)$. Similarly, there can be at most $(k+1)/|L'|$ 
calls to {\bf Alg-SaturatedLines}$(S',k, r)$ which result in $|L'| \ge \frac{1}{12}k^{(1+y_{r})/2}$. By Case (3) of Lemma \ref{lemma9}, the running time of each call is  $\bigoh(|L'| \cdot k(\log^2 k)(\log\log k))$. Therefore, the overall running time for these calls is $\bigoh(k^2(\log^2 k)(\log\log k))$. It follows that the total 
running time of Step~\ref{ker8} is $\bigoh(n\log k+k^2(\log^2 k)(\log\log k)^2)$.

Finally, we bound the running time of Step~\ref{ker13}. This step is executed at most once since $|S'|<2k^2$ holds, which implies that there are no points left in $S$. 
By Theorem \ref{thm:exactfitting}, with probability at least $1-\frac{3}{|S'|^2}>1-\frac{3}{k^4}$, the set
$L'$ computed at Step \ref{ker13} includes all the saturated lines w.r.t.~$S'$. 
By part (2) of Theorem \ref{thm:exactfitting}, the running time of Step \ref{ker13} is $\bigoh(k^2\log^2 k)$. 

Altogether, the running time of {\bf Alg-SaturatedLines} is $\bigoh(n\log k+k^2(\log^2 k)(\log\log k)^2)$. 

Now we analyze the space complexity of the algorithm. First, Step~\ref{special} runs in $\bigoh(1)$ space by Lemma~\ref{lem:specialcase}. Observe that, in each iteration of the {\bf while loop} in Step~\ref{ker2}, the space used by the algorithm is dominated by (1) the space used for storing the sets $H$, $L'$, and $S'$, 
(2) the space used for constructing, updating, and storing the structure $\Gamma_H$, 
and (3) the space utilized to run the two algorithms {\bf Alg-SaturatedLines} and {\bf Alg-RichLines}. Since both $|H|$ and $|L'|$ are $\bigoh(k)$, and since $|S'| \leq 2k^2$, the space used for (1) is $\bigoh(k^2)$.  From the discussion above, 
constructing and updating takes time $\bigoh(k^2\log^2 k)$ and so the space is bounded by $\bigoh(k^2\log^2 k)$. 
Since the size of $\Gamma_H$ is $\bigoh(k^2)$, the space used to store $\Gamma_H$ is $\bigoh(k^2)$~\cite{chan2018}. By Corollary \ref{cor:spacealgrl}, the space used by {\bf Alg-SaturatedLines} is $\bigoh(k^2\log^2 k)$. 
By part (2) of Theorem~\ref{thm:exactfitting}, the time used by {\bf Alg-RichLines} is $\bigoh(k^2\log^2 k)$ and so the space is bounded by $\bigoh(k^2\log^2 k)$. It follows that the space complexity of the algorithm is $\bigoh(k^2\log^2 k)$.
\end{proof}
\fi

\begin{proof}[{\bf Proof of Theorem~\ref{thm:mainlc} stated at the beginning of this section}]
The time and space complexity of the algorithm follow from Lemmas~\ref{lem:specialcase} and~\ref{lem:kerneltime}. We prove its correctness next. The correctness of Step~\ref{special} was proved separately in Lemma~\ref{lem:specialcase}, so we may assume that $k \geq 16$.

Suppose that $(S, k)$ is a no-instance of {\sc Line Cover}. Observe that whenever the algorithm includes a subset $L'$ of lines into the solution $H$ (in Steps~
\ref{ker10} and \ref{ker13}) (and updates $S'$), then the lines in $L'$ are saturated lines, and hence, must be part of \emph{every} solution to the instance $(S, k)$. Therefore, either the algorithm returns an instance in Step~\ref{ker16} that must be a no-instance by the above observation, or returns a (trivial) no-instance in Step \ref{ker9}, \ref{ker15}, or \ref{ker17}. It follows from above that if $(S, k)$ is a no-instance of {\sc Line Cover}, then {\bf Alg-Kernel} returns a no-instance $(S', k')$. This proves part (2) of the theorem.

Suppose now that $(S, k)$ is a yes-instance of {\sc Line Cover}, and hence, that $S$ can be covered by at most $k$ lines. By Step \ref{ker9}, if $L'=\emptyset$, then the algorithm will stop. Thus, Steps \ref{ker7}--\ref{ker10} will be executed at most $k+1$ times.
Consider a single execution of Steps~\ref{ker7}--\ref{ker10}. 
By Lemma \ref{lemma3}, if $S'$ can be covered with at most $k$ saturated lines, then, with probability at least $1-\frac{1}{k^4}$, {\bf Alg-SaturatedLines}$(S', k, r)$
returns a non-empty set $L'$. That is to say, 
{\bf Alg-SaturatedLines}$(S', k, r)$ fails with probability at most $\frac{1}{k^4}$. 
By the union bound, 
{\bf Alg-Kernel}$(S, k)$ fails during the execution of Steps \ref{ker7}--\ref{ker10} 
with probability at most 
$\frac{k+1}{k^4}$. 
At Step \ref{ker13}, by Theorem~\ref{thm:exactfitting}, with probability at least $1-\frac{3}{|S'|^2} > 
1-\frac{3}{k^4}$, 
{\bf Alg-RichLines}$(S',k)$ finds all the saturated lines in $S'$. After that, we have $|S'| \le k^2$. 
By the union bound, with probability at least $1-\frac{k+1}{k^4}-\frac{3}{k^4} > 1-\frac{2}{k^3}$ (since $k \geq 16$), 
{\bf Alg-Kernel}$(S, k)$ returns a kernel $(S', k')$ of $(S, k)$ satisfying $|S'| \le k^2$. This proves part (1) of the theorem. 
\end{proof}

We conclude this section by giving a deterministic kernelization algorithm for {\sc Line Cover}. 
Recall that {\bf Alg-RichLines} is a randomized algorithm for computing all $\lambda$-rich lines and that
Guibas et al.'s algorithm~\cite{Guibas1996} is a deterministic algorithm for the same purpose. 
We can replace {\bf Alg-RichLines} with Guibas et al.'s algorithm~\cite{Guibas1996} in the algorithms {\bf Alg-SaturatedLines}
and {\bf Alg-Kernel} to obtain a deterministic kernelization algorithm from {\bf Alg-Kernel} after this replacement. 
We can optimize the running time of this deterministic algorithm by fine-tuning the lengths of the defined intervals $I_0, \ldots, I_{r+1}$.

\begin{theorem} \ifshort {\rm ($\spadesuit$)} \fi
\label{thm:deterministickernel}
There is a deterministic kernelization algorithm for {\sc Line Cover} that, given an instance $(S, k)$ of {\sc Line Cover}, where $|S|=n$, the algorithm runs in time 
$\bigoh(n\log k+k^3(\log^3 k) \sqrt{\log\log k})$ and computes a kernel $(S', k')$ such that $|S'|\le k^2$.
\end{theorem}

\iflong
\begin{proof}
Let $\epsilon=\frac{\ln\ln\ln k}{\ln k}$. For $i \in \mathbb{N}$, let 
$y_i=\frac{2\ln\ln k}{\ln k}-1+i\epsilon$. Let $r$ be the minimum integer such that $k^{y_r} \ge \frac{k}{(\ln\ln k)^5}$,
and note that $r=\bigoh(\ln k)$. The refined intervals become: $I_0=[\frac{1}{k^{1+y_0}}, 1]$, 
$I_i=[\frac{1}{k^{1+y_i}}, \frac{1}{k^{1+y_{i-1}}})$, for $i=1, 2, \ldots, r,$ and 
$I_{r+1}=[\frac{(k+1)k}{\sigma(\sigma-1)}, \frac{1}{k^{1+y_r}})$. 

Replace {\bf Alg-RichLines} with Guibas et al.'s algorithm~\cite{Guibas1996} in {\bf Alg-SaturatedLines} and {\bf Alg-Kernel},
and call the deterministic kernelization algorithm obtained from {\bf Alg-Kernel} after this replacement {\bf Alg-DetKernel}.
The correctness of {\bf Alg-DetKernel} is obvious. We analyze its running time next. 

When $k\le 15$, {\bf Alg-DetKernel} runs in linear time obviously. We assume henceforth that 
$k\ge 16$. 

First, we analyze the running time of {\bf Alg-SaturatedLines} after replacing {\bf Alg-RichLines} with Guibas et al.'s algorithm~\cite{Guibas1996}. 
In the $i$-th iteration, if $i\le r$, it follows from Guibas et al.'s algorithm~\cite{Guibas1996} that the running time of 
Steps \ref{rich8}-\ref{rich88} is $\bigoh\left( \frac{\sigma^2}{\sigma k^{-(1+y_i)/2}}
\log \frac{\sigma}{\sigma k^{-(1+y_i)/2}}\right)$;  otherwise, the running time of Steps \ref{rich8}-\ref{rich88} is 
$\bigoh(\frac{\sigma^2}{k}\log \frac{\sigma}{k})$.
 Let $T(k)$ be the running time of {\bf Alg-SaturatedLines}. We have the following:

Case (1).  if the algorithm stops at Step \ref{rich-6}, then $L'$ covers at least $\frac{k^2}{3}$ points and 
hence $|L'|\ge 1$. 
We have $T(k) = \bigoh( \frac{\sigma^2}{\sigma k^{-(1+y_0)/2}}
\log \frac{\sigma}{\sigma k^{-(1+y_0)/2}})$. Since $y_0=\frac{2\ln\ln k}{\ln k}-1$, we have 
$k^{1+y_0}=\ln^2 k$. Thus, $T(k)=\bigoh(k^2(\log k)(\log\log k))$ and $T(k)/|L'|=\bigoh(k^2(\log k)(\log\log k))$.

Case (2).  If the algorithm stops at  Step \ref{rich10} with the iterator $i$ of the {\bf for loop} having value $j$, then 
 $|L'|\ge \frac{1}{12r}k^{(1+y_{j-1})/2}$.
 We have $T(k)=\bigoh( \sum_{i=0}^j \frac{\sigma^2}{\sigma k^{-(1+y_i)/2}} \log \frac{\sigma}{\sigma k^{-(1+y_i)/2}})$. Since 
 $r=\bigoh(\ln k)$, we have $T(k)=\bigoh(\sigma (k^{(1+y_j)/2})(\log (k^{(1+y_j)/2})) \cdot \ln k)$
 $=\bigoh(k^2(\log^2 k) (k^{(1+y_j)/2}))$. Recall that, by the definition of $y_j$, we have 
 $y_j-y_{j-1}=\frac{\ln\ln\ln k}{\ln k}$. 
 Hence,  $T(k)/|L'|=\bigoh(k^2(\log^3 k)  (k^{(y_j-y_{j-1})/2}))=\bigoh(k^2(\log^3 k) \sqrt{\log\log k})$. 
 
Case (3).  If the algorithm stops at Step \ref{rich12}, then $|L'| \ge \frac{1}{12}k^{(1+y_{r})/2}$. 
As a consequence, we have $T(k)=\bigoh( \sum_{i=0}^r \frac{\sigma^2}{\sigma k^{-(1+y_i)/2}}$
$\log \frac{\sigma}{\sigma k^{-(1+y_i)/2}}$
$+\frac{\sigma^2}{k}\log \frac{\sigma}{k})=\bigoh(k^2(\log^2 k) \cdot k^{(1+y_r)/2}+k^3\log k)$.
Hence, $T(k)/|L'|=\bigoh(k^2\log^2 k + k^{3-(1+y_r)/2}\log k)$. 
Since $r$ is the minimum integer such that  $k^{y_r} \ge \frac{k}{(\log\log k)^5}$, 
we have 
$k^{y_{r-1}} < \frac{k}{(\ln\ln k)^{5}}$. Thus, $k^{y_r}=k^{y_{r-1}+\epsilon}<\frac{k}{(\ln\ln k)^4}$.
Therefore, $T(k)=\bigoh(k^3(\log k/\log\log k)^2)$ and 
$T(k)/|L'| =\bigoh(k^2\log^2 k)$.

Case (4). Otherwise, the algorithm stops at Step \ref{rich15}, and $T(k)=\bigoh(k^3(\log k/\log\log k)^2)$, 
 which is the same running time as in case (3). 
 
 Taking into account cases (1)--(4), the running time $T(k)=\bigoh(|L'| \cdot k^2(\log^3 k) \sqrt{\log\log k})$ if 
 $|L'|\ge 1$ and $T(k)= \bigoh(k^3(\log k/\log\log k)^2)$ otherwise. 
 
 Now, we upper bound the running time of {\bf Alg-DetKernel}. The 
 running time of Steps \ref{special},  \ref{ker1}, \ref{ker22},  \ref{ker4}, \ref{ker6}, 
 \ref{ker9}, \ref{ker10}, \ref{ker12}, \ref{ker14}, \ref{ker15}, \ref{ker16},  and \ref{ker17} is the same 
 as in {\bf Alg-Kernel}, which is $\bigoh(n\log k+k^2\log^2 k)$.  
 Thus, we mainly focus on analyzing
 the running time of Steps \ref{ker8} and \ref{ker13}. Similar to the proof of Lemma \ref{lem:kerneltime}, 
 at Step \ref{ker8}, at most one call to  {\bf Alg-SaturatedLines} returns $L'=\emptyset$, since {\bf Alg-DetKernel} would stop if the returned set 
 $L'=\emptyset$. Moreover, the running time incurred in such call is 
 $\bigoh(k^3(\log k/\log\log k)^2)$. 
  
 Consider the overall running time of the calls to {\bf Alg-SaturatedLines} in which $|L'|\neq \emptyset$.  
 Since there are at most $k+1$ calls to {\bf Alg-SaturatedLines} such that $|L'|\neq \emptyset$, 
 the overall running time is $\bigoh(|L'| \cdot k^2(\log^3 k) \sqrt{\log\log k} \cdot \frac{k+1}{|L'|})$
 $=\bigoh(k^3(\log^3 k) \sqrt{\log\log k})$.
 
 Step \ref{ker13} applies Guibas et al.'s algorithm~\cite{Guibas1996}, which runs in time 
 $\bigoh(\frac{|S'|^2}{k}\log \frac{|S'|}{k})=\bigoh(k^3\log k)$ since $k^2+1\le |S'|<2k^2$.
 
Altogether, the running time of {\bf Alg-DetKernel} is $\bigoh(n\log k+k^3(\log^3 k) \sqrt{\log\log k})$. 
\end{proof}
\fi

\iflong
\section{Lower Bounds}
\label{sec:lowerbounds}
In this section, we establish time-complexity lower-bound results for {\sc Line Cover} and {\sc Rich Lines} in the algebraic computation trees model~\cite{burgisser}. 

An \emph{algebraic computation tree} is a binary tree whose internal nodes represent arithmetic operations and tests. Thus, a computation step in this model is either an arithmetic
operation or a test, and the computation branches according to the outcome of the operation/test. Each leaf of the tree represents a description of the solution in terms of the input. An algebraic computation tree can be used to depict the execution of an algorithm on a specific input size in a standard way, where for each input that has the specific size, the computation of the algorithm follows a root-leaf path in the tree, performing the operations along this path. The value at a leaf is the algorithm’s output. The running time of the algorithm is the length of the root-leaf path traversed, and the worst-case running time of the algorithm for a specific input size is the depth of the tree. 

The algebraic computation trees model is a more powerful model than the real-RAM model~\cite{adt}, which is the model of computation that is most commonly used to analyze geometric algorithms~\cite{preparata}. The lower-bound results we derive in the algebraic computation trees model apply to the real RAM model as well; for more details see~\cite{adt}.

\subsection{{\sc Line Cover}}
\label{subsec:lclower}
 In this subsection, we are concerned with deriving lower bounds on the time complexity of kernelization algorithms for {\sc Line Cover} in the algebraic computation trees model. To do so, we combine a lower-bound result by Grantson and Levcopoulos~\cite{Grantson2006} on the time complexity of {\sc Line Cover}, derived using the framework introduced by Ben-Or~\cite{benor}, with a result that we prove below connecting the time complexity for solving {\sc Line Cover} to its kernelization time complexity. We remark that, since {\sc Line Cover} is NP-hard~\cite{megiddo} when the parameter $k$ is unbounded, Grantson and Levcopoulos'~\cite{Grantson2006} time complexity lower-bound result for {\sc Line Cover} is interesting only when $k$ is ``small'' relative to the input size, and should be read this way.

\begin{theorem}[Grantson and Levcopoulos~\cite{Grantson2006}]
\label{lowerb}
There exists a constant $c > 0$ such that, for every positive $n, k \in \Nat$ satisfying $k= \bigoh(\sqrt{n})$, {\sc Line Cover} requires time at least $c \cdot n \log k$ in the algebraic computation trees model.
\end{theorem}

We now exploit a folklore connection between kernelization and FPT~\cite{Cygan2015,fptbook} to translate the above time-complexity lower-bound result into a kernelization time-complexity lower-bound result.  We assume that all complexity functions used are proper complexity functions, where by a \emph{proper} complexity function $f: \Nat \longrightarrow \Nat$ we mean a non-decreasing function that, on an input of length $N$, is computable in time $\bigoh(N+f(N))$ and space $\bigoh(f(N))$.

\begin{theorem}
\label{fptker}
Let $Q$ be a parameterized problem in \NP. For any proper complexity function $h$, $Q$ has a kernelization algorithm of running time $\bigoh(h(|x|, k))$, where $(x, k)$ is the input instance to $Q$, if and only if $Q$ can be solved in time $\bigoh(h(|x|, k) + g(k))$ for some proper complexity function $g(k)$.   
\end{theorem}

\begin{proof}
One direction is straightforward: If $Q$ has a kernelization algorithm of running time $\bigoh(h(|x|, k))$ then $Q$ can be solved in time $\bigoh(h(|x|, k) + g(k))$ by applying the kernelization algorithm followed by a brute-force algorithm.

To prove the other direction, suppose that $Q$ can be solved in time $\bigoh(h(|x|, k) + g(k))$. Given an instance $(x, k)$ of $Q$, if $g(k) \leq h(|x|, k)$ then we solve the instance in $\bigoh(h(|x|, k)+g(k))=\bigoh(h(|x|, k))$ time to obtain a trivial kernel. Otherwise, $g(k) > h(|x|, k) \geq |x|$, where the last inequality is true since, w.l.o.g., we may assume that the function $h(|x|, k)$ is at least linear in the input encoding size. Hence, the instance $x$ has size $\bigoh(g(k)+k)$, and hence is kernelized. The above algorithm is a kernelization algorithm, computing a kernel in time $\bigoh(h(|x|, k))$. 
\end{proof}

The corollary below follows from Theorem~\ref{lowerb} and Theorem~\ref{fptker} above:

\begin{corollary}
\label{cor1}
There exists a constant $c > 0$ such that the running time of any kernelization algorithm for {\sc Line Cover} in the algebraic computation trees model is at least $c n \log k$.   
\end{corollary}

\begin{remark}
The above corollary implies that one cannot asymptotically improve on either of the two factors $n$ or $\log k$ in the term $n \log k$. This rules out, for instance, the possibility of a kernelization algorithm that runs in (linear) $O(n)$ time or in $\bigoh(n \log \log k)$ time. 
\end{remark}


\subsection{{\sc Rich Lines}}
\label{subsec:fittinglowerbound}

In this subsection, we derive lower-bound results on the time complexity of {\sc Rich Lines} in the algebraic computation trees model using Ben-Or's framework~\cite{benor}. We briefly describe this framework, and refer to~\cite{benor} for more information.

Ben-Or~\cite{benor} introduced a framework for proving time complexity lower bounds in the algebraic computation trees model. His framework represents the instances of a problem $Q$ as points in the $D$-dimensional Euclidean space. Based on this representation, a solution for $Q$ can be viewed as a sequence of algebraic operations, each splitting the $D$-dimensional space further into regions according to the operation applied. He then showed that if the number of connected components\footnote{A connected component here denotes a set of pints in the space every two of which are connected by a curve whose points belong to the same component.}---viewed as connected regions of the $D$-dimensional space---induced by the set of yes-instances of $Q$ of specific size $n$ is $N$, then any algebraic tree model for $Q$ must have depth at least $\Omega(\log N-D)$, thus proving a lower bound of $\Omega(\log N-D)$ on the time complexity of $Q$ in the algebraic computation trees model. As pointed before, it is well known that any lower-bound result derived using Ben-Or's framework~\cite{benor} in the algebraic computation trees model implies the same lower-bound result in the real RAM model~\cite{adt}. \\

Consider the following problem, which is a variant of the {\sc Element Distinctness} problem~\cite{benor}: 
\begin{quote}
{\sc    Multiset Subset Distinctness} \\
Given a multi-set $A=\{ a_1, a_2, \ldots, a_n \}$ and a positive integer $\lambda$, decide whether $A$ can be partitioned into $n/\lambda$ multi-subsets $A_1, A_2, \ldots, A_{n/\lambda}$, such that each subset $A_i$, where $i \in [n/\lambda]$, contains exactly $\lambda$ identical elements, and no two (distinct) multi-subsets contain identical elements.
\end{quote}
Note that, when $\lambda=1$, {\sc Multiset Subset Distinctness} problem is precisely the {\sc Element Distinctness} problem. Note also that
testing whether an instance $(A, \lambda)$ is a valid instance, and hence, whether $n$ is divisible by $\lambda$, can be trivially done in linear time, using only algebraic
operations.

\begin{theorem} \label{thm:set-partition}
There exists a constant $c > 0$ such that, for every positive $n, \lambda \in \Nat$ such that $\lambda$ divides $n$, {\sc Multiset Subset Distinctness} requires time at least $c \cdot n\log (\frac{n}{\lambda})$ in the algebraic computation trees model. 
\end{theorem}

\begin{proof}
For any fixed $n$ and $\lambda$, the instance $(A, \lambda)$, where $A=(a_1, \ldots, a_n)$, is represented as the point $(a_1, \ldots, a_n, \lambda)$ in the $(n+1)$-dimensional Euclidean space 
$R^{n+1}$. Denote by $W_\lambda^{n+1}$ the set of points in $R^{n+1}$ that corresponds to the set of yes-instances of {\sc Multiset Subset Distinctness}. By Ben-Or's results~\cite[\S4]{benor}, it suffices to show that the number of connected components of $W_{\lambda}^{n+1}$ 
is at least ${n\choose \lambda, \lambda, \ldots, \lambda}=\Theta(\frac{\sqrt{2\pi n}(n/e)^n}{(\sqrt{2\pi \lambda}(\lambda/e)^\lambda)^{n/\lambda}})$~\cite[\S9.6]{stirling}, as this would show that the depth of any algebraic computation tree for {\sc Multiset Subset Distinctness} is at least
$\Omega(\log{n\choose \lambda, \lambda, \ldots, \lambda}) = \Omega(n\log (\frac{n}{\lambda}))$.

Each yes-instance $(A, \lambda)$ of {\sc Multiset Subset Distinctness} corresponds to a mapping $f$ from $[n] \rightarrow [n/\lambda]$ such 
that $f(i)<f(j)$ if and only if $a_i<a_j$, 
and such that $f(i)=f(j)$ if and only if $a_i=a_j$,
and such that for each $j \in [n/\lambda]$: $|\{i\in [n] \mid f(i)=j \}|=\lambda$.
 It is easy to see that the number of such functions $f$ is ${n \choose \lambda, \lambda, \ldots, \lambda}$. 
For each such function $f$, let $W_f$ be the set of yes-instances corresponding to $f$, and let ${\cal W}$
be the set of all subsets $W_{f}$.  It is easy to verify that the sets $W_f$ in $\W$ partition 
$W_{\lambda}^{n+1}$, and that $W_f$ is a connected region/subset in $\mathbb{R}^{n+1}$, as it is the intersection of hyperplanes with a convex set/region.

We prove that, for any two different functions $f$ and $f '$, $W_{f}$ 
and $W_{f'}$ belong to two different connected component of $W_{\lambda}^{n+1}$. Assume 
to the contrary that a point $p\in W_{f}$ and a point $p' \in W_{f'}$ are 
in the same connected component of the set $W_{\lambda}^{n+1}$. Then there is a path $\Pi$ in $W_\lambda^{n+1}$ 
from $p$ to $p'$. This path $\Pi$ can be given in the parametric form as: 
$$ \Pi: \pi(t)=(a_1(t), a_2(t), \ldots, a_n(t), \lambda ), 0\le t\le 1, $$
where $\pi(0)=p$, $\pi(1)=p'$, and each $a_i(t)$, $i \in [n]$ is a continuous function 
of $t$. For an interval $I \subseteq [0,1]$, denote by $\pi(I) = \{ \pi(t) \mid t \in I\}$.

Suppose first that, for each $t\in [0,1]$, there is an open interval $I_t$ containing $t$ such that all points 
in $\pi(I_{t})$ are in the same subset of $\W$. 
Then by the Heine-Borel Theorem \cite{rudin}, we can find a finite set of open intervals
covering $[0,1]$ such that for each such open interval $I_t$, all points in $\pi(I_{t})$ 
are in the same subset of $\W$. This implies that all points on the 
path $\Pi$ are in the same subset of $\W$, contradicting the fact that the subsets 
$W_{f}$ and $W_{f'}$ are disjoint. 

%


Suppose now that there exists a $t_0 \in [0,1]$, where $\pi(t_0)$ is in some $W_{f_1}$, such that for every open interval $I$ containing $t_0$, $\pi(I)$ contains a point 
not in $W_{f_1}$. Since $\W$ is finite, we can construct a sequence $(t)_i$ in $[0, 1]$ converging to $t_0$, and such that, for each $i$, $\pi(t_i)$ belongs to the \emph{same} set 
$W_{f_2} \in \W$, where $f_1 \neq f_2$. Since  $f_1 \neq f_2$, there exist indices $z_1$ and $z_2$ such that 
$z_1 \neq z_2$,  $f_1(z_1) < f_1(z_2)$
and $f_2(z_1) > f_2(z_2)$. 
Consider the sequence of points
$$ \pi(t_r)= ( a_1(t_r), a_2(t_r), \ldots, a_n(t_r), \lambda ), \text{   for } r \ge 1. $$
Since $\pi(t_r)$ approaches $\pi(t_0)$ as $t_r \rightarrow t_0$, 
we must have 
\begin{eqnarray}
&~~~~|a_{z_1}(t_r)-a_{z_1}(t_0)| +|a_{z_2}(t_r)-a_{z_2}(t_0)| \rightarrow 0, \label{J4}
\end{eqnarray} as
$t_r \rightarrow t_0$. 
Recall that $f_1(z_1)<f_1(z_2)$ and $f_2(z_1)> f_2(z_2)$, and hence,
$a_{z_1}(t_0)<a_{z_2}(t_0)$ and $a_{z_1}(t_r)>a_{z_2}(t_r)$. 
It follows that:
\begin{alignat}{2}
&~|a_{z_1}(t_r)-a_{z_1}(t_0)| +|a_{z_2}(t_r)-a_{z_2}(t_0)| \\ 
\ge&~  |(a_{z_1}(t_r)-a_{z_1}(t_0)) -(a_{z_2}(t_r)-a_{z_2}(t_0))| \label{J1} \\  
   \ge&~ | (a_{z_2}(t_0)-a_{z_1}(t_0))+(a_{z_1}(t_r)-a_{z_2}(t_r))  | \label{J2} \\
   \ge&~ | (a_{z_2}(t_0)-a_{z_1}(t_0)) | \label{J3}. 
\end{alignat}
Observing that $a_{z_1}(t_0)$ and 
$a_{z_2}(t_0)$ are fixed, inequality (\ref{J3}) contradicts (\ref{J4}). This completes the proof. 
\end{proof}

Now, we prove a time lower bound $\Omega(n\log \frac{n}{\lambda})$ for the {\sc Rich Lines} problem via a reduction from {\sc Multiset Subset Distinctness} problem. We note that, for a wide range of values of $\lambda$, the number of $\lambda$-rich lines does not exceed the time-complexity lower bound proved below.  In particular, it follows from  Theorem \ref{RichlineBound1} and Lemma \ref{RichlineBound2} that, for $\lambda=\Omega(n^{1/3})$, the number of $\lambda$-rich lines is $O(n)$.

\begin{theorem}
\label{thm:fittinglowerbound}
There exists a constant $c_0 > 0$ such that, for every positive $n, \lambda \in \Nat$, {\sc Rich Lines} requires time at least $c_0 \cdot n\log (\frac{n}{\lambda})$ in the algebraic computation trees model.
\end{theorem}

\begin{proof}
 
We prove the theorem via a Turing-reduction ${\cal T}$ from the {\sc Multiset Subset Distinctness} problem. The theorem would then follow from Theorem~\ref{thm:set-partition}. We first present the reduction.

Given an instance $(A, \lambda)=(a_1, a_2, \ldots, a_n, \lambda)$ of {\sc Multiset Subset Distinctness}, we construct the instance $(P, \lambda)$ of {\sc Rich Lines}, where $P=\{(a_i, i) \mid a_i\in A\}$.
Note that $(P, \lambda)$ can be constructed in $\bigoh(n)$ time. Observe that $(A, \lambda)$ is a yes-instance of {\sc Multiset Subset Distinctness} if and only if 
there are $n/\lambda$ vertical lines that each covers exactly $\lambda$ points of $P$. 
We can solve $(P, \lambda)$ to find the set $L$
of lines induced by $P$ that each covers at least $\lambda$ points. Then, we compute the subset $V$ of vertical lines in $L$ and accept $(A, \lambda)$ if and only if $|V|=n/\lambda$. Let $t(n, \lambda)$ be the time needed to perform this reduction ${\cal T}$.

Now to prove the theorem, we proceed by contradiction. Suppose that no such constant $c_0$ exists, and let $c$ be the universal constant in Theorem~\ref{thm:set-partition}.  Then, for every constant $c'>0$, there exist $n, \lambda \in \Nat$ such that, for \emph{all} input instances of size $n$ and parameter $\lambda$, {\sc Rich Lines} can be solved in time less than $c' \cdot n\log (\frac{n}{\lambda})$.  We observe that, under this assumption, the number of lines in the solution to each of these instances must be less than $c' \cdot n\log (\frac{n}{\lambda})$, otherwise, the running time for solving the instance would necessarily exceed $c' \cdot n\log (\frac{n}{\lambda})$. It is not difficult to see that we can choose a constant $c'>0$ and $n , \lambda \in \Nat$ such that for the specific function $t(n, \lambda)$, where $t(n, \lambda)$ is running time of the reduction ${\cal T}$ given above, we have $t(n, \lambda) + c' \cdot n\log (\frac{n}{\lambda}) < c \cdot n\log (\frac{n}{\lambda})$. Let $n, \lambda$ be the values chosen accordingly. 

Assume first that $\lambda$ divides $n$, and we explain below how the proof can be modified to lift this assumption. Given an instance $(A, \lambda)=(a_1, a_2, \ldots, a_n, \lambda)$ of {\sc Multiset Subset Distinctness}, where $A$ has $n$ elements, we reduce  $(A, \lambda)$ via reduction ${\cal T}$ to an instance $(P, \lambda)$ of {\sc Rich Lines} and solve $(P, \lambda)$ to obtain a solution to $(A,\lambda)$ in time less than $c n\log (\frac{n}{\lambda})$, contradicting Theorem~\ref{thm:set-partition}.

In the case where $\lambda$ does not divide $n$, let $n=r\cdot \lambda + s$, where $s < \lambda$, and let $n'=r \cdot \lambda$.  Observe that the lower bound for {\sc Multiset Subset Distinctness} established in Theorem~\ref{thm:set-partition} holds for the values $n', \lambda$ (since $\lambda$ divides $n'$). Given an instance $(A', \lambda)=(a_1, a_2, \ldots, a_{n'}, \lambda)$ of {\sc Multiset Subset Distinctness}, we construct the instance $(P, \lambda)$ of {\sc Rich Lines}, where $P=P' \cup S$, and $P'=\{(a_i, i) \mid a_i\in A'\}$. The set $S$ contains precisely $s < \lambda$ points and is constructed as follows. We find the smallest element $a_{min} \in A'$, and choose a number $x < a_{min}$. Define $S=\{(x,j) \mid j \in [s]\}$. It is easy to verify that $(A', \lambda)$ is a yes-instance of {\sc Multiset Subset Distinctness} if and only if the number of vertical lines, each containing at least $\lambda$ points of $P$, is $n'/\lambda$. Hence, we can decide $(A', \lambda)$ as explained in the first case above. Note that all the steps involved in the construction of $(P, \lambda)$, including the computation of the number $x$, can be carried out in linear time. Since the constant $c'$ can be chosen to be arbitrary small, it is not difficult to see that we can choose $c'$ and the values $n, \lambda$ such that the running time of the above reduction is less than $c \cdot n'\log (\frac{n'}{\lambda})$, again contradicting Theorem~\ref{thm:set-partition}. Note also that all the operations involved in the above reduction can be equivalently modeled in the algebraic computation trees model~\cite{adt}. This completes the proof. 
\end{proof}

\begin{remark}\rm
\label{rem:fittinglowerbound}
We make a few observations regrading the optimality of {\bf Alg-RichLines} in light of Theorem~\ref{thm:fittinglowerbound}. First, for any value of $\lambda =\Omega(\sqrt{n \ln n})$, by part (4) of Theorem~\ref{thm:exactfitting}, the running time of the randomized algorithm {\bf Alg-RichLines} is $\bigoh(n \ln(\frac{n}{\lambda}))$, which matches the lower bound on the {\sc Rich Lines} problem in the algebraic computation trees model asserted by Theorem~\ref{thm:fittinglowerbound}. Even though the two bounds are w.r.t.~two different models of computation, the fact that the running time of {\bf Alg-RichLines} is $\bigoh(n \ln(\frac{n}{\lambda}))$, which is very close to linear time, and that there is a matching lower bound in the algebraic computations trees model suggests that the randomized algorithm {\bf Alg-RichLines} is near optimal (if not optimal) when $\lambda =\Omega(\sqrt{n \ln n})$. 

Second, for any constant value of $\lambda \geq 1$, and any $n$, one can construct an instance for which the number of lines containing at least $\lambda$ points is $\Omega(n^2)$, as follows. Place $n$ points in a grid of $\lambda$ rows and $\frac{n}{\lambda}$ columns. For any point $p$ on the first row of the grid, there are $\frac{n}{\lambda^2}$ points on the second row such that the line passing through $p$ and one of each one of these points contains $\lambda$ points in the grid; hence in this instance there are at least $n^2/\lambda^3 = \Omega(n^2)$ lines containing at least $\lambda$ points.

This shows that the $\bigoh(n^2)$ running time of {\bf Alg-RichLines} for the case where $\lambda$ is a constant is optimal. We note that the {\sc 3-Points-On-Line} problem is known to be 3-SUM-hard~\cite{king}, suggesting that no $\bigoh(n^{2-\epsilon})$-time algorithm, for any $\epsilon > 0$, exists. However, the {\sc 3-Points-On-Line} problem admits an $o(n^2)$-time randomized algorithms w.r.t.~several computational models~\cite{demaine}. 
\end{remark}
\fi

\section{Conclusion}
\label{sec:conclusion}
\iflong
In this paper, we developed new tools that enabled us to derive upper and lower bounds on the time complexity of fundamental geometric point-line covering problems. The time and space efficiency renders our algorithms amenable to application in massive data processing and analytics.
\fi

Several interesting questions ensue from our work. First, many of the previous algorithms for {\sc Rich Lines} and {\sc Line Cover} can be lifted to higher dimensions (e.g., see~\cite{Guibas1996,things,wangjx}). We believe that it is possible to lift the results in this paper to higher dimensions as well. Second, most of the algorithms we presented are randomized Monte Carlo algorithms. It is interesting to investigate if these algorithms can be derandomized without trading off their performance guarantees by much. Finally, it is interesting to see if the sampling and optimization techniques developed in this paper can be applied to other related problems in computational geometry. We leave all the above questions as directions for future research.

\newpage
\bibliography{ref}

\end{document}